\documentclass[11pt]{article}
\usepackage[utf8]{inputenc}

 \usepackage[T1]{fontenc}

\usepackage[margin=1in]{geometry}
\usepackage{authblk}

\usepackage{amsmath}
\usepackage{amsthm}
\usepackage{amssymb}
\usepackage{thmtools}
\usepackage{thm-restate}
\usepackage{enumitem}
\usepackage{dsfont}
\usepackage{xspace}
\usepackage{comment}
\usepackage{xparse}
\usepackage{tikz}
\usepackage{fixme}
\usepackage[colorlinks]{hyperref}
\usepackage[capitalize]{cleveref}
\usepackage{subcaption}

\setlist[enumerate]{nosep, topsep=0ex}
\setlist[itemize]{nosep, topsep=0ex}

\usepackage{titlesec}
\titlespacing*{\section} {0pt}{2ex}{2ex}
\titlespacing*{\subsection} {0pt}{2ex}{2ex}
\titlespacing*{\subsubsection} {0pt}{2ex}{2ex}

\newtheorem{theorem}{Theorem}[]
\newtheorem{question}{Question}[]

\newtheorem{definition}{Definition}[]
\newtheorem{remark}{Remark}[]

\newtheorem{lemma}{Lemma}[]

\usepackage{algorithm}
\usepackage{algpseudocode}
\usepackage{algorithmicx}  
\usepackage{float}
\usepackage{microtype}
\usepackage{algorithm}
\usepackage{algpseudocode}
\usepackage{algorithmicx}   
\usepackage{floatflt}
\usepackage{graphics}
\usepackage{amsthm}

\title{Optimal Sensitivity Oracle for Steiner Mincut} 

  \author[1]{Koustav Bhanja}

 \affil[1]{Indian Institute of Technology Kanpur, India}
 \affil[ ]{\texttt{kbhanja@cse.iitk.ac.in}}

\date{}

\begin{document}
\begin{titlepage}

\maketitle

\begin{abstract}
    Let $G=(V,E)$ be an undirected weighted graph on $n=|V|$ vertices and $S\subseteq V$ be a Steiner set. Steiner mincut is a well-studied concept, which also provides a generalization to both $(s,t)$-mincut (when $|S|=2$) and global mincut (when $|S|=n$). Here, we address the problem of designing a compact data structure that can efficiently report a Steiner mincut and its capacity after the failure of any edge in $G$; such a data structure is known as a \textit{Sensitivity Oracle} for Steiner mincut.

    In the area of minimum cuts, although many Sensitivity Oracles have been designed in unweighted graphs, however, in weighted graphs, Sensitivity Oracles exist only for $(s,t)$-mincut [Annals of Operations Research 1991, NETWORKS 2019, ICALP 2024], which is just a special case of Steiner mincut. Here, we generalize this result from $|S|=2$ to any arbitrary set $S\subseteq V$, that is, $2 \le |S| \le n$.
    
    We first design an ${\mathcal O}(n^2)$ space Sensitivity Oracle for Steiner mincut by suitably generalizing the approach used for $(s,t)$-mincuts [Annals of Operations Research 1991, NETWORKS 2019].
    However, the main question that arises quite naturally is the following.
    \begin{center}
        \textit{Can we design a Sensitivity Oracle for Steiner mincut that breaks the ${\mathcal O}(n^2)$ bound on space?}
    \end{center}
In this manuscript, we present the following two results that provide an answer to this question.   
\begin{enumerate}
\item \textbf{Sensitivity Oracle:} Assuming the capacity of every edge is known,
\begin{enumerate}
    \item there is an ${\mathcal O}(n)$ space data structure that can report the capacity of Steiner mincut in ${\mathcal O}(1)$ time and
    \item there is an ${\mathcal O}(n(n-|S|+1))$ space data structure that can report a Steiner mincut in ${\mathcal O}(n)$ time
\end{enumerate}
after the failure of any given edge in $G$.
\item \textbf{Lower Bound:} We show that any data structure that, after the failure of any edge, can report a Steiner mincut or its capacity must occupy $\Omega(n^2)$ bits of space in the worst case, irrespective of the size of the Steiner set.  
\end{enumerate}
The lower bound in (2) shows that the assumption in (1) is essential to break the $\Omega(n^2)$ lower bound on space. Sensitivity Oracle in (1.b) occupies only subquadratic, that is ${\mathcal O}(n^{1+\epsilon})$, space if $|S|=n-n^{\epsilon}+1$, for every $\epsilon\in [0,1)$.  For $|S|=n-k$ for any constant $k\ge 0$, it occupies only ${\mathcal O}(n)$ space. So, we also present the first Sensitivity Oracle occupying ${\mathcal O}(n)$ space for global mincut.
In addition, we are able to match the existing best-known bounds on both space and query time for $(s,t)$-mincut [Annals of Operations Research 1991, NETWORKS 2019] in undirected graphs.
\end{abstract} 
\end{titlepage}

\tableofcontents{}
\pagebreak
\pagenumbering{arabic}
\section{Introduction}
\label{sec : introduction}
 In the real world, networks (graphs) are often subject to the failure of edges and vertices due to a variety of factors, such as physical damage, interference, or other disruptions. This can lead to changes in the solution to several graph problems. While these failures can happen at any location in the network at any time, they are typically short-lived. Naturally, it requires us to have compact data structures that can efficiently report the solution to the given graph problem (without computing from scratch) once any failure has occurred. Such data structures are known as \textit{Sensitivity Oracles} for several graph problems. There exist elegant Sensitivity Oracles for many fundamental graph problems, such as shortest paths  \cite{baswana2020approximate, bilo2023approximate}, reachability  \cite{italiano2021planar, choudhary2016optimal}, traversals \cite{DBLP:conf/podc/Parter15, baswana2019dynamic}, etc.
 
 The minimum cut of a graph is also a fundamental concept of graph theory. Moreover, it has a variety of practical applications in the real world \cite{DBLP:books/daglib/0069809}.  
 Designing Sensitivity Oracles for various minimum cuts of a graph has been an emerging field of research for the past few decades \cite{DBLP:journals/talg/BaswanaBP23, DBLP:journals/anor/ChengH91, DBLP:journals/siamcomp/DinitzV00, dinitz19952,
 baswana2022mincut, bhanja2024minimum+, baswana2022sensitivity, DBLP:conf/icalp/BaswanaB24}. There are two well-known mincuts of a graph. They are global mincut and (s,t)-mincut.
Here, we design the first Sensitivity Oracle for global mincut in undirected weighted graphs that can handle the failure of any edge.
The concept of Steiner mincut is also well-studied in the area of minimum cuts \cite{dinitz1994connectivity, dinitz1995locally, DBLP:journals/siamcomp/DinitzV00, baswana2022sensitivity, cole2003fast, he2024cactus, DBLP:journals/corr/abs-1912-11103}; moreover, it has global mincut, as well as $(s,t)$-mincut, as just a special corner case. In this article, as our main result, we present the first \textit{Sensitivity Oracle for Steiner mincut} for handling the failure of any edge in undirected weighted graphs.
Interestingly, our result bridges the gap between the two extreme scenarios of Steiner mincut while matching their bounds, namely,   $(s,t)$-mincut \cite{DBLP:journals/networks/AusielloFLR19, DBLP:journals/anor/ChengH91} and global mincut (designed in this article). In addition, it also provides the first generalization from unweighted graphs \cite{baswana2022sensitivity, dinitz1994connectivity, dinitz1995locally, DBLP:journals/siamcomp/DinitzV00} to weighted graphs. 

Let $G=(V,E)$ be an undirected graph on $n=|V|$ vertices and $m=|E|$ edges with non-negative real values assigned as the capacity to edges. We denote the capacity of an edge $e$ by $w(e)$. Let $S\subseteq V$ be a \textit{Steiner set} of $G$ such that $|S|\ge 2$. A vertex $s$ is called a \textit{Steiner vertex} if $s\in S$; otherwise, $s$ is called a \textit{nonSteiner vertex}. 
\begin{definition} [Steiner cut] \label{def : steiner cut}
    A nonempty set $C\subset V$ is said to be a \textit{Steiner cut} if there is at least one pair of Steiner vertices $s,s'$ such that $s\in C$ and $s'\notin C$. 
\end{definition}
For $S=V$, a Steiner cut is a \textit{(global) cut}. Similarly, for $S=\{s,t\}$, a Steiner cut is an \textit{$(s,t)$-cut}. 
A cut $C$ is said to \textit{separate} a pair of vertices $u,v$ if $u\in C$ and $v\in \overline{C}=V\setminus C$ or vice versa. An edge $e=(u,v)$ is said to \textit{contribute} to a cut $C$ if $C$ separates endpoints $u,v$ of $e$. The \textit{capacity of a cut} $C$, denoted by $c(C)$, is the sum of capacities of all contributing edges of $C$. A Steiner cut of the least capacity is known as the \textit{Steiner mincut}, denoted by $S$-mincut. Let $\lambda_S$ be the capacity of $S$-mincut. 
The problem of designing a Sensitivity Oracle for $S$-mincut for handling the failure of any edge is defined as follows.
\begin{definition} [single edge Sensitivity Oracle for Steiner mincut]
    For any graph $G$, a single edge Sensitivity Oracle for Steiner mincut is a compact data structure that can efficiently report a Steiner mincut and its capacity after the failure of any edge in $G$.
\end{definition}
For unweighed graphs, there exist single edge Sensitivity Oracles for global mincut \cite{dinitz1976structure}, $(s,t)$-mincut \cite{DBLP:journals/mp/PicardQ80, DBLP:journals/talg/BaswanaBP23}, and Steiner mincut \cite{dinitz1994connectivity, dinitz1995locally, baswana2022sensitivity}. Unfortunately, for weighted graphs, in the area of minimum cuts, the only existing results are single edge Sensitivity Oracles for $(s,t)$-mincut \cite{DBLP:journals/networks/AusielloFLR19, DBLP:journals/anor/ChengH91, DBLP:conf/icalp/BaswanaB24}.
For undirected weighted graphs, Ausiello et al. \cite{DBLP:journals/networks/AusielloFLR19}, exploiting the Ancestor tree data structure of Cheng and Hu \cite{DBLP:journals/anor/ChengH91}, designed the first single edge Sensitivity Oracle for $(s,t)$-mincut. Their Sensitivity Oracle occupies ${\mathcal O}(n^2)$ space.
After the failure of any edge, it can report an $(s,t)$-mincut $C$ and its capacity in ${\mathcal O}(|C|)$ and ${\mathcal O}(1)$ time, respectively.  
Recently, Baswana and Bhanja \cite{DBLP:conf/icalp/BaswanaB24} complemented this result by showing that $\Omega(n^2\log~n)$ bits of space is required in the worst case, irrespective of the query time. 
 
For Steiner mincuts, it follows from the above discussion that the existing Sensitivity Oracles are either for undirected unweighted graphs or only for a special case, when $|S|=2$, in weighted graphs. Therefore, to provide a generalization of these results to any Steiner set, the following is an important question to raise.
\begin{center}
    \textit{Does there exist a single edge Sensitivity Oracle for $S$-mincut in undirected weighted graphs?}
\end{center}
\noindent
We show that the approach taken by Ausiello et al. \cite{DBLP:journals/networks/AusielloFLR19} can be generalized from $S=\{s,t\}$ to any set $S\subseteq V$. This answers the above-mentioned question in the affirmative and leads to the following result.   
\begin{theorem} \label{thm : single edge sensitivity oracle using cheng and hu}
    For any undirected weighted graph $G$ on $n=|V|$ vertices, for every Steiner set $S$, there exists an ${\mathcal O}(n^2)$ space data structure that, after the failure of any edge in $G$, can report an $S$-mincut $C$ and its capacity in ${\mathcal O}(|C|)$ time and ${\mathcal O}(1)$ time respectively.
\end{theorem}
The space and query time of the Sensitivity Oracle in Theorem \ref{thm : single edge sensitivity oracle using cheng and hu} match with the existing optimal results for $(s,t)$-mincut \cite{DBLP:journals/anor/ChengH91, DBLP:journals/networks/AusielloFLR19, DBLP:conf/icalp/BaswanaB24}. 
The lower bound of $\Omega(n^2\log~n)$ bits of space in \cite{DBLP:conf/icalp/BaswanaB24} is only for $|S|=2$. However, to the best of our knowledge, no lower bound is known for any $|S|>2$.
Therefore, the main question that we address in this article arises quite naturally as follows.
\begin{question} \label{question : 1}
    For undirected weighted graphs, does there exist a single edge Sensitivity Oracle for $S$-mincut that breaks the quadratic bound on space and still achieves optimal query time if $|S|>2$?
\end{question}
\subsection{Our Results} \label{sec : Our results}
A Sensitivity Oracle in a weighted graph addresses queries in a more generic way \cite{DBLP:conf/icalp/BaswanaB24}.  Given any edge $e$ and any value $\Delta$ satisfying $\Delta \ge 0$, the aim is to efficiently report the solution of a given problem after reducing the capacity of edge $e$ by $\Delta$.  In this generic setting, we first design an ${\mathcal O}(n)$ space single edge Sensitivity Oracle for global mincut that also achieves optimal query time. Now, in order to bridge the gap between the two extreme scenarios of Steiner set ($|S|=n$ and $|S|=2$) while matching their bounds, we present our main result that breaks the ${\mathcal O}(n^2)$ space bound of Theorem \ref{thm : single edge sensitivity oracle using cheng and hu}, and answers Question \ref{question : 1} in the affirmative. 

\begin{theorem}[Sensitivity Oracle for Steiner Mincut] \label{thm : main result}
    Let $G=(V,E)$ be an undirected weighted graph on $n=|V|$ vertices and $m=|E|$ edges. For any Steiner set $S$ of $G$,  
    \begin{enumerate}
        \item there is an ${\mathcal O}(n)$ space rooted tree ${\mathcal T}(G)$ that, given any edge $e\in E$ and any value $\Delta$ satisfying $0\le \Delta\le w(e)$, can report the capacity of $S$-mincut in ${\mathcal O}(1)$ time after reducing the capacity of edge $e$ by $\Delta$ and 
        
        \item  there is an ${\mathcal O}(n(n-|S|+1))$ space data structure ${\mathcal F}(G)$ that, given any edge $e\in E$ and any value $\Delta$ satisfying $0\le \Delta\le w(e)$, can report an $S$-mincut $C$ in ${\mathcal O}(|C|)$ time after reducing the capacity of edge $e$ by $\Delta$.
    \end{enumerate}
\end{theorem}
    For any $\epsilon\in [0,1)$, the space occupied by the single edge Sensitivity Oracle for $S$-mincut in Theorem \ref{thm : main result}(2) is subquadratic, that is ${\mathcal O}(n^{1+\epsilon})$, for $|S|=n-n^{\epsilon}+1$. Moreover, it approaches to ${\mathcal O}(n)$ as $|S|$ tends to $n$. In particular, for $|S|=n-k$, for any constant $k\ge 0$, it occupies only ${\mathcal O}(n)$ space. 
    
    Observe that our results in Theorem \ref{thm : main result} interestingly match the bounds on both space and query time for the two extreme scenarios of the Steiner set. On one extreme ($|S|=n$), it occupies ${\mathcal O}(n)$ space for global mincut. On the other extreme ($|S|=2$), it occupies ${\mathcal O}(n^2)$ space, which match the best-known existing results for $(s,t)$-mincut \cite{DBLP:journals/networks/AusielloFLR19, DBLP:journals/anor/ChengH91, DBLP:conf/icalp/BaswanaB24}. Finally, notice that the time taken by our Sensitivity Oracle to answer any query is also worst-case optimal.

We also provide lower bounds on both space and query time of Sensitivity Oracles for $S$-mincut. Our first lower bound is for reporting the capacity of $S$-mincut and our second lower bound is for reporting an $S$-mincut. 

\begin{theorem} [Lower Bound for Reporting Capacity] \label{thm : lower bound on reporting capacity}
    Let $D$ be any data structure that can report the capacity of Steiner mincut after the failure of any edge for undirected weighted graphs on $n$ vertices. Data structure $D$ must occupy $\Omega(n^2\log~n)$ bits of space in the worst case, irrespective of the query time and the size of the Steiner set. 
\end{theorem}
For reporting the capacity of $S$-mincut, Theorem \ref{thm : lower bound on reporting capacity} provides a generalization of the existing lower bound on both space and time for $(s,t)$-mincut by Baswana and Bhanja \cite{DBLP:conf/icalp/BaswanaB24}. However, for reporting an $S$-mincut, no lower bound on space or query time for single edge Sensitivity Oracle was known till date, even for the two extreme scenarios of Steiner set. So, the following theorem is the first lower bound for reporting an $S$-mincut after the failure of any edge.   
\begin{theorem}[Lower Bound for Reporting Cut] \label{thm : lower bound on reporting cut}
    Let $D$ be any data structure that can report a Steiner mincut $C$ in ${\mathcal O}(|C|)$ time after the failure of any edge for undirected weighted graphs on $n$ vertices. Data structure $D$ must occupy $\Omega(n^2)$ bits of space in the worst case, irrespective of the size of Steiner set.
\end{theorem}

\begin{remark}
    It is assumed in Theorem \ref{thm : main result} that the query edge $e$ is present in $G$ and the change in capacity (that is, $\Delta$) provided with the query is at most $w(e)$.
    So, the lower bounds of $\Omega(n^2)$ bits of space in Theorem \ref{thm : lower bound on reporting capacity} and Theorem \ref{thm : lower bound on reporting cut} do not violate the sub-quadratic space data structures in Theorem \ref{thm : main result}. Moreover, the assumption in Theorem \ref{thm : main result} seems practically justified. This is because, as discussed in \cite{DBLP:conf/icalp/BaswanaB24}, in the real world, the capacity of an edge reduces only if the edge actually exists in the graph, and furthermore, it can reduce by a value at most the capacity of the edge.
\end{remark}
\subsection{Related Works}
In the seminal works by Dinitz and Vainshtein \cite{dinitz1994connectivity, dinitz1995locally, DBLP:journals/siamcomp/DinitzV00}, they designed an ${\mathcal O}(\min\{n\lambda_S,m\})$ space data structure, known as \textit{Connectivity Carcass}, for storing all $S$-mincuts of an unweighted undirected graph. 
It can report an $S$-mincut in ${\mathcal O}(m)$ time and its capacity in ${\mathcal O}(1)$ time. 
Baswana and Pandey \cite{baswana2022sensitivity}, using Connectivity Carcass as the foundation, designed an ${\mathcal O}(n)$ space single edge Sensitivity Oracle for $S$-mincut in undirected unweighted graphs that also reports an $S$-mincut in ${\mathcal O}(n)$ time. Their result matches the bounds on both space and time for the existing result on the two extreme scenarios of $S$-mincut, namely, $(s,t)$-mincut \cite{DBLP:journals/mp/PicardQ80} and global mincut \cite{dinitz1976structure}. The result in \cite{baswana2022sensitivity} also acts as the foundation of single edge Sensitivity Oracles for all-pairs mincut \cite{baswana2022sensitivity}

For directed weighted graphs, Baswana and Bhanja \cite{DBLP:conf/icalp/BaswanaB24} presented a single edge Sensitivity Oracle for $(s,t)$-mincuts that matches both space and query time of the undirected weighted graph results \cite{DBLP:journals/anor/ChengH91, DBLP:journals/networks/AusielloFLR19}.

\noindent
Providing a generalization from the two extreme scenarios of the Steiner set ($S=V$ and $|S|=2$) is also addressed for various problems, namely, computing Steiner mincut \cite{dinitz1994connectivity, DBLP:journals/siamcomp/DinitzV00, cole2003fast, he2024cactus, DBLP:journals/corr/abs-1912-11103}, 
Steiner connectivity augmentation and splitting-off \cite{cen2023steiner}, construction of a cactus graph for Steiner mincuts \cite{DBLP:journals/siamcomp/DinitzV00, he2024cactus}.

\subsection{Organization of the Article}
This article is organized as follows. Section \ref{sec : Preliminaries} contains the basic preliminaries. We first construct an ${\mathcal O}(n^2)$ space single edge Sensitivity Oracle for Steiner mincut in Section \ref{sec : simple sensitivity oracle}. In Section \ref{sec : simple data structure}, we design an ${\mathcal O}(n)$ space single edge Sensitivity Oracle for reporting only the capacity of Steiner mincut. A linear space single edge Sensitivity Oracle for global mincut is designed in Section \ref{sec : global mincut}. Our main result on the subquadratic space single edge Sensitivity Oracle for Steiner mincut is developed in Section \ref{sec : weighted edge failure oracle}. The lower bound results are given in Section \ref{sec : lower bound}. Finally, we conclude in Section \ref{sec : conclusion} with couple of open problems.

\section{Preliminaries} \label{sec : Preliminaries}
In this section, we define a set of basic notations and properties of cuts.
Let $G\setminus \{e\}$ denote the graph obtained from $G$ after the removal of edge $e$.
We now define the concept of crossing cuts,  introduced by Dinitz, Karzanov, and Lomonosov \cite{dinitz1976structure}. 

\begin{definition}[crossing cuts]
    A pair of cuts $C,C'$ in $G$ is said to be crossing if each of the four sets $C\cap C'$, $C\setminus C'$, $C'\setminus C$, and $\overline{C\cup C'}$ is nonempty. 
\end{definition}
The following concept of mincut for an edge and vital edges are to be used crucially in the construction of our data structure.
\begin{definition} [Mincut for an edge]
    A Steiner cut $C$ is said to be a mincut for an edge $e$ if $e$ contributes to $C$ and $c(C)\le c(C')$ for every Steiner cut $C'$ in which $e$ contributes. 
\label{def : mincut for an edge}
\end{definition}
\begin{definition} [Vital Edge] \label{def : vital edge}
    Let $e$ be an edge and $C$ be a mincut for edge $e$. Edge $e$ is said to be a vital edge if its removal reduces the capacity of Steiner mincut, that is, $c(C)-w(e)<\lambda_S$.
\end{definition}
We now define a \textit{special} mincut for an edge.
 \begin{definition}[Nearest mincut for an edge] \label{def : nearest mincut}
    A mincut $C$ for an edge $e=(x,y)\in E$ with $x\in C$ is said to be a nearest mincut for $e$ if there is no mincut $C'$ for $e$ such that $x\in C'$ and $C'\subset C$. The set of all nearest mincuts for an edge $e$ is denoted by $N(e)$.
\end{definition}
\begin{lemma}[Sub-modularity of Cuts (Problem 48(a,b) in \cite{Lovasz-book})] \label{submodularity of cuts}
    For any two sets $A,B\subset V$,
    \begin{enumerate}
        \item $c(A)+c(B)\ge c(A\cap B)+c(A\cup B)$ and
        \item $c(A)+c(B)\ge c(A\setminus B)+c(B\setminus A)$.
    \end{enumerate}
\end{lemma}
For undirected weighted graphs, Gomory and Hu \cite{gomory1961multi} designed the following tree structure, which is widely known as \textsc{Gomory Hu Tree}.
\begin{theorem} [\textsc{Gomory Hu Tree} \cite{gomory1961multi}] \label{thm : Gomory Hu tree}
    For any undirected weighted graph $G=(V,E)$ on $n=|V|$ vertices, there is an ${\mathcal O}(n)$ space undirected weighted tree ${\mathcal T}_{GH}$ on vertex set $V$ that satisfies the following property. Let $u,v$ be any pair of vertices in $G$. A cut of the least capacity separating $u,v$ in ${\mathcal T}_{GH}$ is also a cut of the least capacity separating $u,v$ in $G$. 
    Moreover, ${\mathcal T}_{GH}$ can report a cut $C$ of the least capacity separating $u,v$ in ${\mathcal O}(|C|)$ time and its capacity in ${\mathcal O}(1)$ time. 
\end{theorem}
\begin{definition} [Laminar family of cuts] \label{def : laminar family}
    A set of cuts ${\mathcal L}$ is said to form a laminar family if, for any pair of cuts $C_1,C_2\in {\mathcal L}$, exactly one of the three is true -- $C_1\cap C_2$ is an empty set, $C_1\subseteq C_2$, and $C_2\subseteq C_1$.
\end{definition}
\paragraph*{A rooted tree ${\mathcal T}_{\mathcal L}$ representing a laminar family ${\mathcal L}$:} For any given laminar family ${\mathcal L}$ of cuts, we can construct a rooted tree ${\mathcal T}_{\mathcal L}$ as follows. 
 Let $x$ be any vertex in $G$.
Let $\phi_{\mathcal L}(x)$ denote the unique node in ${\mathcal T}_{\mathcal L}$ to which vertex $x$ is mapped and let $\textsc{SubTree}(x)$ denote the set of all vertices mapped to the subtree rooted at $\phi_{\mathcal L}(x)$ (including $\phi_{\mathcal L}(x)$) in ${\mathcal T}_{\mathcal L}$. 
The set \textsc{SubTree}($x$) defines the unique minimal cut in ${\mathcal L}$ that contains $x$. If $\phi_{\mathcal L}(x)$ is a child of the root node in ${\mathcal T}_{\mathcal L}$, then  \textsc{SubTree}($x$) is a maximal cut in ${\mathcal L}$. For any pair of vertices $x,y$ in $G$, let $C_1=\textsc{SubTree}($x$)$ and $C_2=\textsc{SubTree}($y$)$. Then, 
$\phi_{\mathcal L}(x)$ is a child of $\phi_{\mathcal L}(y)$ in ${\mathcal T}_{\mathcal L}$ if and only if $C_1$ is a maximal proper subset of $C_2$ in ${\mathcal L}$. A vertex in $G$ is mapped to the root node of ${\mathcal T}_{\mathcal L}$ if no cut in ${\mathcal L}$ contains it.  This leads to the following lemma.
\begin{lemma} \label{lem : tree for a laminar family}
    For any laminar family ${\mathcal L}$ of cuts, there exists an ${\mathcal O}(n)$ space rooted tree ${\mathcal T}_{\mathcal L}$ such that a cut $C\in {\mathcal L}$ if and only if there exists a node $\mu$ (except root node) of ${\mathcal T}_{\mathcal L}$ and $C$ is the set of vertices mapped to the subtree rooted at $\mu$ (including node $\mu$).
\end{lemma}


\section{An ${\mathcal O}(n^2)$ Space Sensitivity Oracle for Steiner Mincut} \label{sec : simple sensitivity oracle}
In this section, we first provide the limitations of the previous results in unweighted graphs. 
Later, we design an ${\mathcal O}(n^2)$ space single edge Sensitivity Oracle for $S$-mincut. 
\paragraph*{Limitations of the existing results} 
For unweighted graphs, the following property is used crucially to design every existing single edge Sensitivity Oracle. \\
\noindent
\textsc{Property $P_1$:} \textit{Failure of an edge $e$ reduces the capacity of $S$-mincut if and only if edge $e$ contributes to an $S$-mincut.}\\ Dinitz and Vainshtein \cite{dinitz1994connectivity, dinitz1995locally, DBLP:journals/siamcomp/DinitzV00} designed the following quotient graph, known as the flesh graph, of $G$. Flesh graph is obtained by contracting every pair of vertices in $G$ that are not separated by any $S$-mincut. The construction ensures that every pair of vertices in flesh is separated by an $S$-mincut of $G$. 
Every vertex in $G$ is mapped to a unique vertex in flesh. Therefore, the endpoints of any edge $e$ are mapped to different vertices in flesh if and only if failure of $e$ reduces capacity of $S$-mincut. Thus, by \textsc{Property $P_1$}, storing the ${\mathcal O}(n)$ space mapping of vertices of $G$ to the vertices of flesh is sufficient to answer the capacity of $S$-mincut in ${\mathcal O}(1)$ time after the failure of any edge. In addition, 
Dinitz and Vainshtein \cite{dinitz1994connectivity, dinitz1995locally, DBLP:journals/siamcomp/DinitzV00} showed that the flesh graph can be used to report an $S$-mincut after the failure of any edge in $G$ in ${\mathcal O}(m)$ time. 

\noindent
Flesh graph is one of the three components of Connectivity Carcass designed by Dinitz and Vainshtein \cite{dinitz1994connectivity, dinitz1995locally, DBLP:journals/siamcomp/DinitzV00}; the other two components are \textit{Skeleton} and \textit{Projection mapping}.
Recently, Baswana and Pandey \cite{baswana2022sensitivity}, exploiting the properties of all the three components of Connectivity Carcass established an \textit{ordering} among the vertices of flesh graph. By using \textsc{Property $P_1$}, they showed that this ordering, along with Skeleton and Projection mapping, can be used to design an ${\mathcal O}(n)$ space single edge Sensitivity Oracle for $S$-mincut in unweighted graphs. This can report an $S$-mincut in ${\mathcal O}(n)$ time.

Unfortunately, for weighted graphs, it is easy to observe that multiple edges can exist that do not contribute to any $S$-mincut but failure of each of them reduces the capacity of $S$-mincut. Hence, in weighted graphs, \textsc{Property $P_1$} no longer holds. So, the existing structures are not suitable for handling the failure of weighted edges. It requires us to explore the structure of mincuts for every edge whose both endpoints belong to the same node of flesh graph. Moreover, the capacity of mincut for these edges can be quite \textit{large} compared to the capacity of $S$-mincut. 
\subsection*{Sensitivity Oracle for Steiner Mincut: ${\mathcal O}(n^2)$ Space} 
We now give a proof of Theorem \ref{thm : single edge sensitivity oracle using cheng and hu} by designing an ${\mathcal O}(n^2)$ space single edge Sensitivity Oracle for $S$-mincut. Let $F$ be any arbitrary real-valued function defined on cuts. Cheng and Hu \cite{DBLP:journals/anor/ChengH91} presented the following result. There is an ${\mathcal O}(n^2)$ space data structure, known as Ancestor tree, that, given any pair of vertices $u$ and $v$, reports a cut $C$ of the least capacity ($F$-value) separating $u,v$ in ${\mathcal O}(|C|)$ time and the capacity of $C$ in ${\mathcal O}(1)$ time.

In order to design Ancestor tree for Steiner cuts, similar to $(s,t)$-mincuts given by Ausiello et al. \cite{DBLP:journals/networks/AusielloFLR19}, we define function $F$ for Steiner cuts as follows. 
\begin{equation} \label{eq : s,t cuts}
\text{For a set $C\subset V$,~}F(C)=\begin{cases}
            c(C), \text{ if $C$ is a Steiner cut} \\
            \infty, \quad \text{otherwise.}
     \end{cases}
\end{equation}
Let $e=(x,y)$ be any failed edge. Ancestor tree can report a cut $C$ of the least capacity separating $x$ and $y$ in ${\mathcal O}(|C|)$ time and its capacity in ${\mathcal O}(1)$ time. By Equation \ref{eq : s,t cuts}, $C$ is also a Steiner cut separating $x$ and $y$. Therefore, by Definition \ref{def : mincut for an edge}, $C$ is a mincut for edge $e$. Hence, the new capacity of $S$-mincut is either $c(C)-w(e)$ or remains $\lambda_S$ if $c(C)-w(e)\ge \lambda_S$. By storing the capacities of all edges of $G$, we can determine whether $c(C)-w(e)< \lambda_S$ in ${\mathcal O}(1)$ time. If the capacity of $S$-mincut reduces, then we can report $C$ in ${\mathcal O}(|C|)$ time; otherwise report an $S$-mincut $C_m$ in ${\mathcal O}(|C_m|)$ time.  
This completes the proof of Theorem \ref{thm : single edge sensitivity oracle using cheng and hu}.
\section{A Sensitivity Oracle for Reporting Capacity of Steiner Mincut} \label{sec : simple data structure}
In this section, we address the problem of reporting the capacity of $S$-mincut after reducing the capacity of an edge $e\in E$ by a value $\Delta$ satisfying $0 \le \Delta\le w(e)$. We denote this query by $\textsc{cap}(e,\Delta)$. Observe that a trivial data structure for answering query $\textsc{cap}$ occupies ${\mathcal O}(m)$ space if we store the capacity of mincut for each vital edge in $G$. For $|S|=2$ in directed weighted graphs, Baswana and Bhanja \cite{DBLP:conf/icalp/BaswanaB24} designed an ${\mathcal O}(n)$ space data structure that implicitly stores all vital edges for $(s,t)$-mincut and the capacity of their mincuts. 
We extend their approach from vital edges to all edges in undirected weighted graphs 
in order to establish the following.  For any Steiner set $S$, with $2\le |S|\le n$, there exists an ${\mathcal O}(n)$ space data structure that can answer query $\textsc{cap}$ in ${\mathcal O}(1)$ time.

Let ${\mathcal E}\subseteq E$ and $V({\mathcal E})$ denote the smallest set of vertices such that, for each edge $(u,v)\in {\mathcal E}$, both $u$ and $v$ belongs to $V({\mathcal E})$. We first design an ${\mathcal O}(|V({\mathcal E})|)$ space rooted full binary tree for answering query $\textsc{cap}$ for all edges in ${\mathcal E}$. In Section \ref{sec : weighted edge failure oracle}, this construction also helps in designing a compact structure for reporting mincuts for a \textit{special} subset of edges. Later in this section, we show an extension to ${\mathcal O}(n)$ space rooted full binary tree for answering query $\textsc{cap}$ for all edges in $E$.

Let $C(e)$ denote a mincut for an edge $e$. Note that $C(e)$ is a Steiner cut as well (Definition \ref{def : mincut for an edge}). We say that an edge $e$ belongs to a set $U\subset V$ if both endpoints of $e$ belong to $U$.
Suppose $C(e)$ is a mincut for an edge $e$ belonging to $V({\mathcal E})$ such that, for every other edge $e'\in V({\mathcal E})$, $c(C(e))\le c(C(e'))$.
Let $e'$ be an edge from $V({\mathcal E})$. If $e'$ contributes to $C(e)$, it follows from the selection of edge $e$ that $C(e)$ is a Steiner cut of the least capacity to which $e'$ contributes. Hence, $C(e)$ is a mincut for edge $e'$ as well. This ensures that $C(e)$ partitions the set of all edges belonging to $V({\mathcal E})$ into three sets -- edges of $V({\mathcal E})$ belonging to $C(e)\cap V({\mathcal E})$, edges of $V({\mathcal E})$ belonging to $\overline{C(e)}\cap V({\mathcal E})$, and edges of $V({\mathcal E})$ that contribute to $C(e)$. This leads to a recursive procedure (Algorithm \ref{alg : hierarchy tree}) for the construction of a tree ${\mathcal T}$. Each internal node $\mu$ of tree ${\mathcal T}$ has three fields -- $(i)$ $\mu.\text{cap}$ stores the capacity of mincut for the selected edge at $\mu$, $(ii)$ $\mu.\textsc{left}$ points to the left child of $\mu$, and $(iii)$ $\mu.\textsc{right}$ points to the right child of $\mu$. Each vertex $u\in V({\mathcal E})$ is mapped to a leaf node of ${\mathcal T}$, denoted by ${\mathbb L}(u)$.  We invoke Algorithm \ref{alg : hierarchy tree} with $U=V({\mathcal E})$.

\begin{algorithm}[h]
\caption{Construction of Tree ${\mathcal T}$}
\label{alg : hierarchy tree}
\begin{algorithmic}[1]
\Procedure{\textsc{SteinerTreeConstruction}$(U)$}{}
    \State Create a node $\nu$; 
    \If {there is no edge that belongs to $U$}{ 
        \For{each vertex $x\in U$} {${\mathbb L}(x)\gets \nu$};
        \EndFor}
    \Else
       \State Select an edge $e\in U$ such that $c(C(e))\le c(C(e'))$ $\forall~\text{edge}~ {e'}\in U$;
        \State Assign $\nu.\text{cap}\gets c(C(e))$;
        \State $\nu.\textsc{left} \gets$ \textsc{SteinerTreeConstruction}($U\cap C(e)$);
        \State $\nu.\textsc{right} \gets$ \textsc{SteinerTreeConstruction}($U\cap \overline{C(e)}$);
    \EndIf
    \State \Return $\nu$;
\EndProcedure
\end{algorithmic}
\end{algorithm}
\noindent
Observe that tree ${\mathcal T}$ resulting from Algorithm \ref{alg : hierarchy tree} is a full binary tree. There are ${\mathcal O}(|V({\mathcal E})|)$ leaf nodes. So, the space occupied by the tree is ${\mathcal O}(|V({\mathcal E})|)$.

\paragraph*{Answering Query $\textsc{cap}(e=(x,y),\Delta):$} Suppose edge $e$ belongs to ${\mathcal E}$. Let $\mu$ be the \textsc{lca} of  ${\mathbb L}(x)$ and ${\mathbb L}(y)$. It follows from the construction of tree ${\mathcal T}$ that field $\mu.\text{cap}$ at node $\mu$ in ${\mathcal T}$ stores the capacity of mincut for edge $e$. Therefore, if $\mu.\text{cap}-\Delta<\lambda_S$, then we report $\mu.\text{cap}$ as the new capacity of S-mincut; otherwise, the capacity of S-mincut does not change. It leads to the following lemma.

\begin{lemma} \label{lem : tree with k endpoints of vital edges}
    Let $G=(V,E)$ be an undirected weighted graph on $n=|V|$ vertices. For any Steiner set $S\subseteq V$ and a set of edges ${\mathcal E}\subseteq E$, there is an ${\mathcal O}(|V({\mathcal E})|)$ space full binary tree ${\mathcal T}_{\mathcal E}$ that, given any edge $e\in {\mathcal E}$ and any value $\Delta$ satisfying $0\le \Delta\le w(e)$, can report the capacity of S-mincut in ${\mathcal O}(1)$ time after reducing the capacity of edge $e$ by $\Delta$. 
\end{lemma}
We now answer query $\textsc{cap}(e,\Delta)$ where edge $e\in E$. Observe that edge $e$ in query $\textsc{cap}$ can be either a vital or a nonvital edge. In order to determine whether an edge is vital or not, we design a full binary tree ${\mathcal T}_E$ by invoking Algorithm \ref{alg : hierarchy tree} with $U=V$  since ${\mathcal E}=E$. Let us denote the tree by ${\mathcal T}(G)$. By Lemma \ref{lem : tree with k endpoints of vital edges}, the size of tree ${\mathcal T}(G)$ is ${\mathcal O}(n)$. It is now easy to observe that an edge $e$ is a vital edge in graph $G$ if and only if the capacity of the Steiner mincut in graph $G\setminus \{e\}$ is $\mu.\text{cap}-w(e)<\lambda_S$, where node $\mu$ is the \textsc{lca}$({\mathbb L}(x), {\mathbb L}(y))$. This leads to the following lemma.
\begin{lemma} \label{lem : report cap and determine vitality}
    Let $G=(V,E)$ be an undirected weighted graph on $n=|V|$ vertices. For any Steiner set $S\subseteq V$, there is an ${\mathcal O}(n)$ space full binary tree ${\mathcal T}(G)$ that, given any edge $e\in E$ and any value $\Delta$ satisfying $0\le \Delta\le w(e)$, can report the capacity of $S$-mincut in ${\mathcal O}(1)$ time after reducing the capacity of edge $e$  by $\Delta$. 
\end{lemma}
Lemma \ref{lem : report cap and determine vitality} completes the proof of Theorem \ref{thm : main result}(1).

\section{An ${\mathcal O}(n)$ Space Sensitivity Oracle for Global Mincut} \label{sec : global mincut}
The well-known $(s,t)$-mincut is one extreme scenario of $S$-mincut when $|S|=2$. In weighted graphs, designing a single edge Sensitivity Oracle for $(s,t)$-mincut has been addressed quite extensively \cite{DBLP:journals/networks/AusielloFLR19, DBLP:journals/anor/ChengH91, DBLP:conf/icalp/BaswanaB24}. Moreover, each of them occupies ${\mathcal O}(n^2)$ space. However, to this day, no nontrivial single edge Sensitivity Oracle exists for global mincut, which is the other extreme scenario of $S$-mincut when $|S|=n$. We now present the first single edge Sensitivity Oracle for global mincut that occupies only ${\mathcal O}(n)$ space and achieves optimal query time.

Let $\lambda_V$ be the capacity of global mincut. Given any edge $e$, we want to determine the capacity of mincut for edge $e$ for Steiner set $S=V$.
Observe that, for $S=V$, every cut in the graph is a Steiner cut (or global cut). Exploiting this insight, we can state the following interesting relation between global mincut and all-pairs mincuts (or $(u,v)$-mincut, for every $u,v\in V$). 
\begin{lemma} \label{lem : relation to gomory hu tree}
    For an edge $(u,v)$, $C$ is a cut of the least capacity that separates $u,v$ if and only if $C$ is a mincut for edge $(u,v)$.
\end{lemma}
By Theorem \ref{thm : Gomory Hu tree}, for every pair of vertices $u,v$ in $G$, \textsc{Gomory Hu Tree} (Theorem \ref{thm : Gomory Hu tree}) stores a cut of the least capacity separating $u,v$. By Lemma \ref{lem : relation to gomory hu tree}, it follows that, for $S=V$, \textsc{Gomory Hu Tree} stores a mincut for every edge in $G$. Hence, it acts as a single edge Sensitivity Oracle for global mincut and
 can report a mincut $C$ for any given edge $e$ in ${\mathcal O}(|C|)$ time and its capacity in ${\mathcal O}(1)$ time. Therefore, after reducing $w(e)$ by a value $\Delta$, if $c(C)-\Delta<\lambda_V$, we can report a global mincut and its capacity optimally using ${\mathcal O}(n)$ space. This leads to the following result. 
\begin{theorem} [Sensitivity Oracle for Global Mincut] \label{thm : global mincut}
    For any undirected weighted graph $G=(V,E)$ on $n=|V|$ vertices, there is an ${\mathcal O}(n)$ space data structure that, given any edge $e$ in $G$ and any value $\Delta$ satisfying $0\le \Delta\le w(e)$, can report the capacity of global mincut in ${\mathcal O}(1)$ time and a global mincut $C$ in ${\mathcal O}(|C|)$ time  after reducing the capacity of edge $e$ by $\Delta$. 
\end{theorem}
Now, for both extreme scenarios of Steiner mincuts, we have a single edge Sensitivity Oracle. Interestingly, the Sensitivity Oracle for global mincut achieves better than quadratic space. 
Therefore, the question that arises is how to generalize these results to any Steiner set.


\section{Sensitivity Oracle for Steiner Mincut: Breaking Quadratic Bound} \label{sec : weighted edge failure oracle}
In this section, we address the problem of reporting an $S$-mincut after reducing the capacity of any given edge $e$ by any given value $\Delta$ satisfying $0<\Delta\le w(e)$. We denote this query by $\textsc{cut}(e,\Delta)$. Our objective is to design a data structure that breaks ${\mathcal O}(n^2)$ bound on space for efficiently answering query \textsc{cut}, if $|S|>2$. A simple data structure can be designed by augmenting tree ${\mathcal T}(G)$ in Theorem \ref{thm : main result}(1) as follows. For each internal node $\mu$ of tree ${\mathcal T}(G)$, Algorithm \ref{alg : hierarchy tree} selects an edge $e$ in Step 7 and stores the capacity of mincut $C(e)$ for edge $e$ in $\mu.\text{cap}$. Observe that if we augment node $\mu$ with $C(e)$, then it helps in answering query $\textsc{cut}$ as well. However, the augmented tree occupies ${\mathcal O}(n^2)$ space, which defeats our objective. 

For global mincut ($S=V$), observe that \textsc{Gomory Hu Tree} essentially acts as a data structure that stores at least one mincut for every edge quite compactly. 
To design a more compact data structure for answering query \textsc{cut} for $S$-mincut compared to Theorem \ref{thm : single edge sensitivity oracle using cheng and hu}, we take an approach of designing a data structure that can compactly store at least one mincut for every edge.  

We begin by a \textit{classification} of all edges of graph $G$. This classification not only helps in combining the approaches taken for $(s,t)$-mincut and global mincut but also provides a way to design a compact data structure for efficiently answering query \textsc{cut} for any Steiner set $S$. 

\paragraph*{A Classification of All Edges:} An edge $e_1$ in $G$ belongs to 
\begin{itemize}
    \item Type-1 if both endpoints of $e_1$ belong to $V\setminus S$.
    \item  Type-2 if both endpoints of $e_1$ belong to $S$.
    \item Type-3 if exactly one endpoint of $e_1$ belongs to $S$.
\end{itemize}

\noindent
Given any edge $e_1$, we can classify $e_1$ into one of the above-mentioned three types in ${\mathcal O}(1)$ time using sets $V$ and $S$. Note that edges from Type-1 allow us to extend an approach for $(s,t)$-mincut given by Baswana and Bhanja \cite{DBLP:conf/icalp/BaswanaB24}. Similarly, edges from Type-2 help in extending the approach used for global mincut.  However, the \textbf{main challenge} arises in designing a data structure for compactly storing a mincut for all edges from Type-3. We now design a compact data structure for efficiently answering query \textsc{cut} for each type of edges separately.



\subsection{An ${\mathcal O}((n-|S|)n)$ Space Data Structure for All Edges from Type-1}
 In this section, we design a data structure for answering \textsc{cut} for all edges from Type-1.  Each edge from Type-1 has both endpoints in the set $V\setminus S$. The number of vertices in $V\setminus S$ is $n-|S|$. Therefore, trivially storing a mincut for every edge would occupy ${\mathcal O}((n-|S|)^2n)$ space, which is ${\mathcal O}(n^3)$ for $|S|=k$ for any constant $k\ge 2$. Exploiting the fact that the number of distinct endpoints of all edges from Type-1 is at most $n-|S|$, we design an ${\mathcal O}(n(n-|S|))$ space data structure for all edges from Type-1 using Algorithm \ref{alg : hierarchy tree} and Lemma \ref{lem : tree with k endpoints of vital edges} as follows.
 
 Let $E_1$ be the set of all edges from Type-1. It follows from Lemma \ref{lem : tree with k endpoints of vital edges} that, using Algorithm \ref{alg : hierarchy tree}, it is possible to design a rooted full binary tree ${\mathcal T}_{E_1}$ occupying ${\mathcal O}(n-|S|)$ space for answering query $\textsc{cap}(e,\Delta)$ when edge $e$ is from Type-1. 
 We augment each internal node $\mu$ of ${\mathcal T}_{E_1}$ with a mincut for the edge selected by Algorithm \ref{alg : hierarchy tree} (in Step 7) while processing node $\mu$. The resulting structure occupies ${\mathcal O}((n-|S|)n)$ space and acts as a data structure for answering query $\textsc{cut}$ for all edges from Type-1. Hence the following lemma holds.
\begin{lemma} [Sensitivity Oracle for Type-1 Edges] \label{lem : type 1 result}
    For any Steiner set $S\subseteq V$, there is an ${\mathcal O}((n-|S|)n)$ space data structure that, given any edge $e$ from Type-1 and any value $\Delta$ satisfying $0\le \Delta\le w(e)$, can report an $S$-mincut $C$ in ${\mathcal O}(|C|)$ time after reducing the capacity of edge $e$ by $\Delta$. 
\end{lemma}
\subsection{An ${\mathcal O}(n)$ Space Data Structure for All Edges from Type-2}
In this section, we design a data structure for answering query \textsc{cut} for all edges from Type-2. For each edge from Type-2, both endpoints are Steiner vertices. So, the number of distinct endpoints of edges from Type-2 can be at most $|S|$. 
Trivially, storing a mincut for every edge from Type-2 would occupy ${\mathcal O}(|S|^2n)$ space, which is ${\mathcal O}(n^3)$ if $|S|={\mathcal O}(n)$. 
In a similar way as designing ${\mathcal T}_{E_1}$ for all edges from Type-1 (Lemma \ref{lem : type 1 result}), by using Lemma \ref{lem : tree with k endpoints of vital edges} and Algorithm \ref{alg : hierarchy tree}, it is possible to design an ${\mathcal O}(|S|n)$ space data structure for answering query $\textsc{cut}$ for all edges from Type-2. Unfortunately, it defeats our objective because, for $|S|=n$ or global mincuts, it occupies ${\mathcal O}(n^2)$ space. Interestingly, by exploiting the fact that both the endpoints of every edge from Type-2 are Steiner vertices, we are able to show that a \textsc{Gomory Hu Tree} of graph $G$ is sufficient for answering query \textsc{cut} for all edges from Type-2.
\begin{lemma} \label{lem : gomory hu tree type 2}
    \textsc{Gomory Hu Tree} of $G$ stores a mincut for every edge from Type-2.
\end{lemma}
\begin{proof}
    By Theorem \ref{thm : Gomory Hu tree}, \textsc{Gomory Hu Tree} stores a cut of the least capacity separating every pair of vertices in $G$. Let $(u,v)$ be any edge from Type-2. Suppose $C$ is a cut of the least capacity separating $u$ and $v$ in $\textsc{Gomory Hu Tree}$ of $G$. So, edge $(u,v)$ is contributing to $C$. By definition of Type-2 edges, both $u$ and $v$ are Steiner vertices. Therefore, $C$ is a Steiner cut in which edge $(u,v)$ is contributing.  It follows from Theorem \ref{thm : Gomory Hu tree} that $C$ is also a cut of the least capacity in $G$ that separates $u$ and $v$. So, $C$ is also a Steiner cut of the least capacity in which edge $(u,v)$ is contributing. Hence, $C$ is a mincut for edge $(u,v)$.
\end{proof}
Let $e=(u,v)$ be any edge from Type-2. 
By Theorem \ref{thm : Gomory Hu tree}, \textsc{Gomory Hu Tree} of $G$ can determine in ${\mathcal O}(1)$ time the capacity of a cut $C$ of the least capacity in $G$ that separates $u$ and $v$. By Lemma \ref{lem : gomory hu tree type 2}, $C$ is a mincut for edge $e$. So, again by Theorem \ref{thm : Gomory Hu tree}, \textsc{Gomory Hu Tree} can be used to report mincut $C$ for edge $e$ in ${\mathcal O}(|C|)$ time. This completes the proof of the following lemma.
\begin{lemma}[Sensitivity Oracle for Type-2 Edges] \label{lem : type 2 result}
    For any Steiner set $S\subseteq V$, there is an ${\mathcal O}(n)$ space data structure that, given any edge $e$ from Type-2 and any value $\Delta$ satisfying $0\le \Delta\le w(e)$, can report an $S$-mincut $C$ in ${\mathcal O}(|C|)$ time after reducing the capacity of edge $e$ by $\Delta$. 
\end{lemma}
For global mincut or $S=V$, both endpoints of every edge are Steiner vertices. Therefore, Theorem \ref{thm : global mincut} can also be seen as a corollary of Lemma \ref{lem : type 2 result}.

\subsection{An ${\mathcal O}((n-|S|)n)$ Space Data Structure for All Edges from Type-3}
In this section, the objective is to design a data structure for answering query $\textsc{cut}$ for all edges from Type-3. 
Observe that the size of the smallest set of vertices that contains all the endpoints of all edges from Type-3 can be $\Omega(n)$. Therefore, using Lemma \ref{lem : tree with k endpoints of vital edges} and Algorithm \ref{alg : hierarchy tree}, we can have an ${\mathcal O}(n^2)$ space data structure, which is no way better than the trivial data structure for answering query \textsc{cut} (Theorem \ref{thm : single edge sensitivity oracle using cheng and hu}). Now, each edge from Type-3 has exactly one nonSteiner endpoint. So, unlike edges from Type-2, Lemma \ref{lem : gomory hu tree type 2} no longer holds for edges from Type-3.
This shows the limitations of the approaches taken so far in designing a data structure for answering query $\textsc{cut}$. Trivially storing a mincut for every edge from Type-3 requires ${\mathcal O}((n-|S|)n|S|)$ space. For $|S|=\frac{n}{k}$, any constant $k\ge 2$, it occupies ${\mathcal O}(n^3)$ space. Interestingly, we present a data structure occupying only ${\mathcal O}((n-|S|)n)$ space for answering query \textsc{cut} for all edges from Type-3.

For any edge $e_1=(x,u)$ with $x\in S$ from Type-3, without loss of generality, we assume that any mincut $C$ for edge $e_1$  contains the Steiner vertex $x$, otherwise consider $\overline{C}$. 
Note that the set of global mincuts and $(s,t)$-mincuts are closed under both intersection and union. This property was crucially exploited in designing a compact structure for storing them \cite{dinitz1976structure, DBLP:journals/mp/PicardQ80}. 
To design a compact structure for storing a mincut for every edge from Type-3, we also explore the relation between a pair of mincuts for edges from Type-3. Let $A$ and $B$ be mincuts for edges $e_1$ and $e_2$ from Type-3, respectively. Unfortunately, it turns out that if $A$ crosses $B$, then it is quite possible that neither $A\cap B$ nor $A\cup B$ is a mincut for $e_1$ or $e_2$ even if both are Steiner cuts (refer to Figure \ref{fig : type-3 edges}($i$)). This shows that mincuts for edges from Type-3 are not closed under intersection or union. To overcome this hurdle, we first present a partitioning of the set of edges from Type-3 based on the nonSteiner vertices as follows.


\begin{figure}
 \centering
    \includegraphics[width=0.9\textwidth]{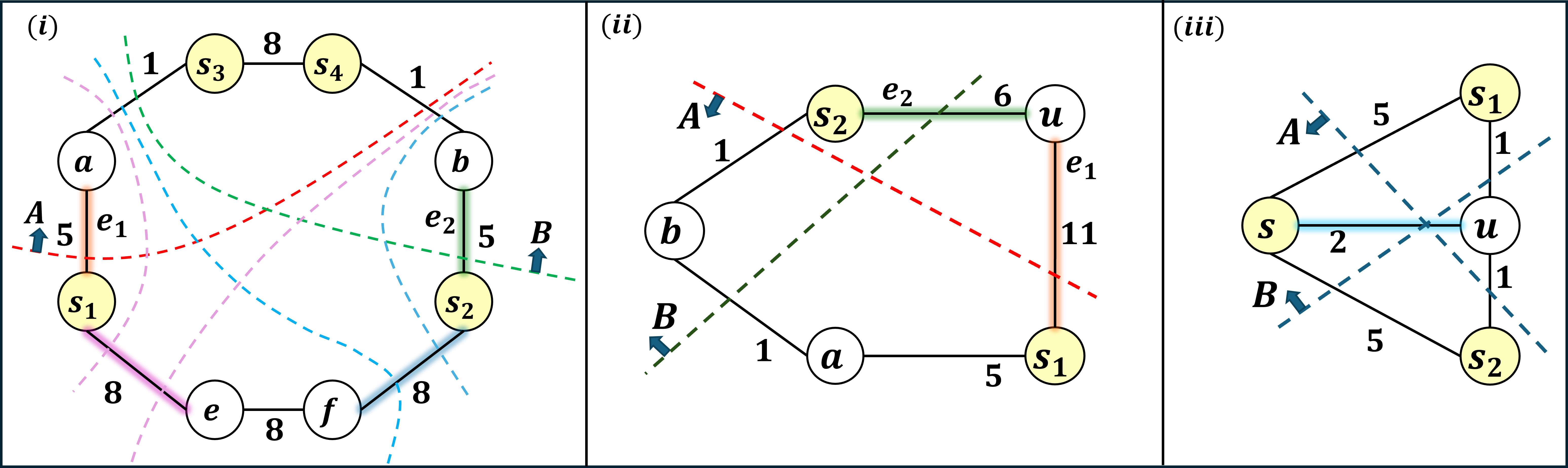} 
   \caption{Yellow vertices are Steiner vertices. A mincut for an edge is represented by the same color. $(i)$ Mincuts $A,B$ are for edges $e_1=(s_1,a)$ and $e_2=(s_2,b)$. Observe that $A\cap B$ and $A\cup B$ are Steiner cuts but not mincuts for edges $e_1,e_2$. Moreover, cuts $A\setminus B$ and $B\setminus A$, in which edges $e_1$ and $e_2$ are contributing, are not even Steiner cuts. $(ii)$ Edges $e_1$ and $e_2$ are vital and from Type-3($u$). Mincuts $A$ and $B$ for edges $e_1$ and $e_2$ are crossing, but $A\cap B$ and $A\cup B$ are not Steiner cuts. $(iii)$ Edge $(s,u)$ is from Type-3 and $N((s,u))=\{A,B\}$. }
  \label{fig : type-3 edges} 
\end{figure}

Let $V'\subseteq V\setminus S$ be the smallest set of nonSteiner vertices such that every edge from Type-3 has endpoint in $V'$. Let $u$ be any vertex from $V'$. Let Type-3($u$) be the set that contains all edges from Type-3 having $u$ as one of the two endpoints. We aim to design an ${\mathcal O}(n)$ space data structure that can report a mincut for each edge from Type-3($u$). This is because storing an ${\mathcal O}(n)$ space data structure for every nonSteiner vertex of $G$ would lead to an ${\mathcal O}((n-|S|)n)$ space data structure.

In order to design an ${\mathcal O}(n)$ space data structure for edges from Type-3($u$), we consider the set of nearest mincuts (Definition \ref{def : nearest mincut}) for edges from Type-3($u$). The following lemma provides a strong reason behind the use of nearest mincuts for edges from Type-3($u$). 
\begin{lemma}[\textsc{Disjoint Property}] \label{lem : disjointness of nearest mincuts}
    Let $C\in N(e_1)$ and $C'\in N(e_2)$ such that $e_1=(x,u)$, $e_2=(x',u)$ are edges from Type-3($u$). Then, $x'\notin C$ and  $x\notin C'$ if and only if $C\cap C'=\emptyset$. 
\end{lemma}
\begin{proof}
    Suppose $C\cap C'\neq \emptyset$. Since $x\notin C'$ and $x'\notin C$, so we have $x\in C\setminus C'$  and $x'\in C'\setminus C$. Evidently, $C\setminus C'$, as well as $C'\setminus C$, is a Steiner cut. It is given that $C$ is the nearest mincut for edge $(x,u)$ and $x\in C\setminus C'$. This implies that $c(C\setminus C')> c(C)$. It follows from sub-modularity of cuts (Lemma \ref{submodularity of cuts}$(2)$) that $c(C'\setminus C)<c(C')$. Therefore, we get a Steiner cut $C'\setminus C$ of capacity strictly less than $c(C')$ and edge $(x',u)$ is a contributing edge of Steiner cut $C'\setminus C$, a contradiction. 

    The proof of the converse part is immediate. 
\end{proof}
Let $e_1=(x,u)$ and $e_2=(x',u)$ be any pair of edges from Type-3($u$). Let $C$ be a nearest mincut for $e_1$ and $C'$ be a nearest mincut for $e_2$. Lemma \ref{lem : disjointness of nearest mincuts} essentially states that if $e_2$ contributes to $C'\setminus C$ and $e_1$ contributes to $C\setminus C'$, then $C$ must not cross $C'$. Now, the problem arises when one of the two edges $e_1,e_2$ is contributing to a nearest mincut for the other edge. Firstly, there might exist multiple nearest mincuts for an edge (refer to Figure \ref{fig : type-3 edges}($iii$)). It seems quite possible that an edge, say $e_2$, is contributing to one nearest mincut for $e_1$ and is not contributing to another nearest mincut for $e_1$. Secondly, the union of a pair of nearest mincuts for an edge $e_1$ from Type-3($u$) might not even be a Steiner cut if they cross (refer to Figure \ref{fig : type-3 edges}($iii$)). 
Hence, the union of them can have a capacity strictly less than the capacity of mincut for $e$. So, it seems that the nearest mincuts for edges from Type-3($u$) appear quite \textit{arbitrarily}. It might not be possible to have an ${\mathcal O}(n)$ space structure for storing them. Interestingly, we are able to circumvent all the above challenges as follows. 

Observe that we are interested in only those edges from Type-3($u$) whose failure reduces $S$-mincut. They are the set of all vital edges that belong to Type-3($u$), denoted by VitType-3($u$). By exploiting \textit{vitality} of edges from VitType-3($u$), we establish the following crucial insight for any pair of crossing mincuts for edges from VitType-3($u$). Interestingly, the following result holds even if the union of a pair of mincuts for a pair of edges from VitType-3($u$) is not always a Steiner cut (refer to Figure \ref{fig : type-3 edges}($ii$)). 

 \begin{lemma} [\textsc{Property of Intersection}]\label{lem : main lemma for weighted edge} 
     Let $C_1$ and $C_2$ be mincuts for edges $e_1=(x_1,u)$ and $e_2=(x_2,u)$ from VitType-3$(u)$ respectively. Steiner vertex $x_2$ is present in $C_1$ if and only if $C_1\cap C_2$ is a mincut for edge $e_2$. 
 \end{lemma}
      \begin{figure}
 \centering
    \includegraphics[width=0.8\textwidth]{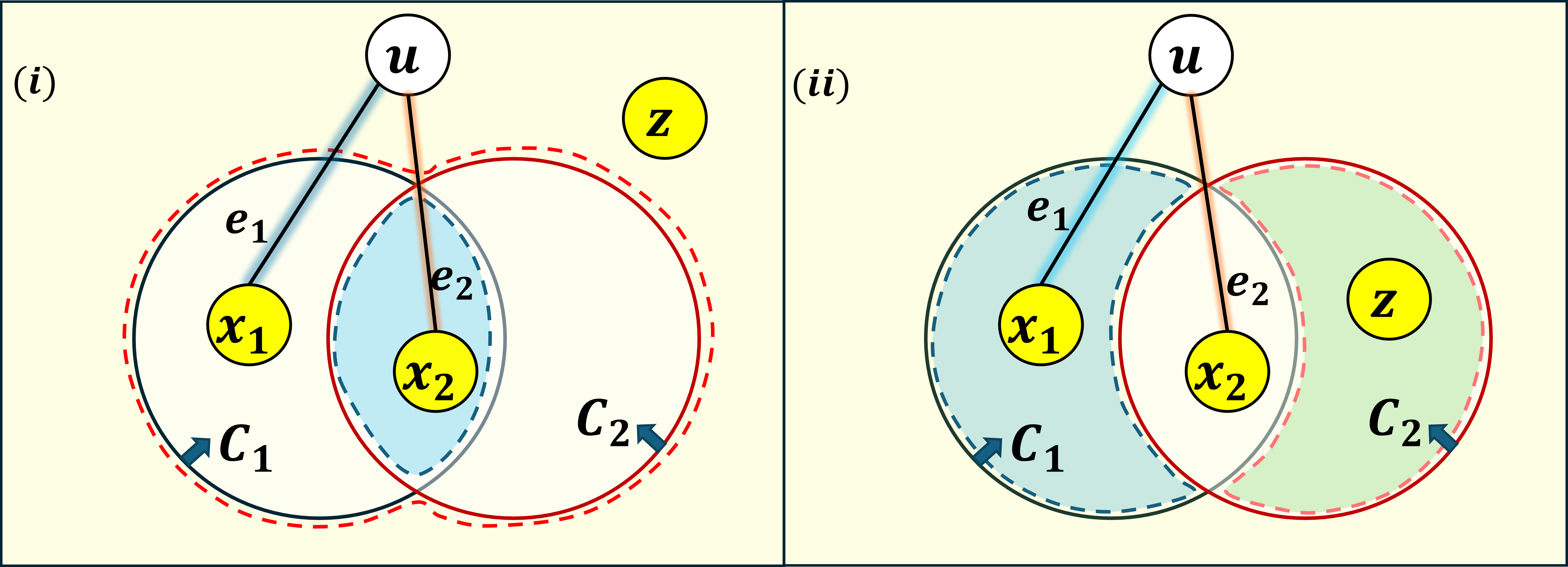} 
   \caption{Illustration of the proof of Lemma \ref{lem : main lemma for weighted edge}. $(i)$ There is a  Steiner vertex $z$ in $\overline{C_1\cup C_2}$. The red dashed cut shows cut $C_1\cup C_2$ and the blue dashed cut (blue region) shows cut $C_1\cap C_2$. $(ii)$ There is a Steiner vertex $z$ in $C_2\setminus C_1$. Blue dashed cut (blue region) is $C_1\setminus C_2$. Similarly, red dashed cut (light green region) is $C_2\setminus C_1$. }
  \label{fig : main lemma proof weighted edge} 
\end{figure}
 \begin{proof}
    Suppose Steiner vertex $x_2$ is present in $C_1$. Since $u$ is a common endpoint of both edges $e_1,e_2$, we have $u\notin C_1\cup C_2$. $C_1$ is a mincut for edge $e_1$, so, $x_1\in C_1$. Now, there are two possibilities -- either $(1)$ $x_1\notin C_2$ or $(2)$ $x_1\in C_2$. We now establish each case separately.\\

    \noindent
    \textbf{Case $\mathbf{1}$.} Suppose $x_1\notin C_2$, in other words, $x_1\in C_1\setminus C_2$. It implies that $C_1\setminus C_2$ is nonempty. We consider that $C_2$ is not a subset of $C_1$; otherwise, $C_1\cap C_2=C_2$ is a mincut for $e_2$ and the lemma holds. So, $C_2\setminus C_1$ is also nonempty. Let $c(C_1)=\lambda_1$ and $c(C_2)=\lambda_2$. Since $x_2\in C_1\cap C_2$, edge $e_2$ is contributing to $C_1\cap C_2$. Therefore, $c(C_2)\le c(C_1\cap C_2)$; otherwise, $C_2$ is not a mincut for edge $e_2$. Since $C_1$ is a Steiner cut, there must be a Steiner vertex $z$ such that $z\notin C_1$. Based on the position of $z$ with respect to cut $C_2$, observe that $z$ appears either $(1.1)$ in $\overline{C_1\cup C_2}$ or  $(1.2)$ in $C_2\setminus C_1$ (refer to Figure \ref{fig : main lemma proof weighted edge}). \\

\noindent
    \textbf{Case $\mathbf{1.1}$.} Suppose $z\in \overline{C_1\cup C_2}$ (refer to Figure \ref{fig : main lemma proof weighted edge}$(i)$). Observe that $C_1\cup C_2$ is a Steiner cut in which edge $e_1$ is contributing. So, the capacity of $C_1\cup C_2$ has to be at least $\lambda_1$. By sub-modularity of cuts (Lemma \ref{submodularity of cuts}$(1)$), $c(C_1\cap C_2)+c(C_1\cup C_2)\le \lambda_1+\lambda_2$. It follows that $c(C_1\cap C_2)\le \lambda_2$. Since $C_1\cap C_2$ is a Steiner cut in which edge $e_2$ is contributing, therefore, the capacity of $C_1\cap C_2$ is exactly $\lambda_2$. Hence $C_1\cap C_2$ is a mincut for edge $e_2$.\\
   
   \noindent
   \textbf{Case 1.2.} We show that this case does not arise. Assume to the contrary that $z\in C_2\setminus C_1$ (refer to Figure \ref{fig : main lemma proof weighted edge}$(ii)$). Here, we crucially exploit the fact that edge $e_2$ is a vital edge from Type-3($u$). Let us now consider graph $G\setminus \{e_2\}$. Since $e_2$ is a vital edge in $G$, therefore, in graph $G\setminus \{e_2\}$, the capacity of $S$-mincut is $\lambda_2-w(e_2)$ and $C_2$ is an $S$-mincut. In $G$, edge $e_2$ is also a contributing edge of $C_1$. Therefore, the capacity of $C_1$ in $G\setminus \{e_2\}$ is $\lambda_1-w(e_2)$. Without causing any ambiguity, let us denote the capacity of any cut $A$ in $G\setminus \{e_2\}$ by $c(A)$. By sub-modularity of cuts (Lemma \ref{submodularity of cuts}$(2)$), in graph $G\setminus \{e_2\}$, we have $c(C_1\setminus C_2)+c(C_2\setminus C_1)\le \lambda_1+\lambda_2-2w(e_2)$. Recall that $x_1 \in C_1\setminus C_2$ and $z\in C_2\setminus C_1$. So, both $C_1\setminus C_2$ and $C_2\setminus C_1$ are Steiner cuts in graph $G\setminus \{e_2\}$. Therefore, the capacity of $C_2\setminus C_1$ is at least $\lambda_2-w(e_2)$. It follows that the capacity of $C_1\setminus C_2$ in $G\setminus\{e_2\}$ is at most $\lambda_1-w(e_2)$. We now obtain graph $G$ by adding edge $e_2$ to graph $G\setminus \{e_2\}$. Observe that edge $e_2$ does not contribute to Steiner cut $C_1\setminus C_2$. Therefore, the capacity of cut $C_1\setminus C_2$ remains the same in graph $G$, which is at most $\lambda_1-w(e_2)$. Since $e_2$ is a vital edge, so, $w(e_2)>0$. This implies that we have $\lambda_1-w(e_2)<\lambda_1$. Therefore, for cut $C_1\setminus C_2$ in $G$, $C_1\setminus C_2$ is a Steiner cut and has a capacity that is strictly less than $\lambda_1$. Moreover, edge $e_1$ is contributing to $C_1\setminus C_2$. So, $C_1$ is not a mincut for edge $e_1$, a contradiction. \\
   
       \noindent
   \textbf{Case ${\mathbf{2.}}$} In this case, we have $x_1\in C_2$. Since $C_1$ and $C_2$ both are Steiner cuts, observe that either (2.1) there is at least one Steiner vertex $z$ in $\overline{C_1\cup C_2}$ or (2.2) there exists a pair of Steiner vertices $z_1,z_2$ such that $z_1\in C_1\setminus C_2$ and $z_2\in C_2\setminus C_1$. The proof of case (2.1) is along similar lines to the proof of case (1.1). So, let us consider case (2.2). 
   Edges $e_1$ and $e_2$ are contributing to both $C_1$ and $C_2$. It implies that $c(C_1)=c(C_2)$. Let $c(C_1)$ (or $c(C_2)$) be $\lambda$. Let us consider graph $G\setminus \{e_2\}$. Since $e_2$ is a vital edge, the capacity of $S$-mincut is $\lambda-w(e_2)$. The capacity of cuts $C_1$ and $C_2$ in $G\setminus \{e_2\}$ is $\lambda-w(e_2)$ since $e_2$ contributes to both of them. Without causing any ambiguity, let us denote the capacity of any cut $A$ in $G\setminus \{e_2\}$ by $c(A)$. By sub-modularity of cuts (Lemma \ref{submodularity of cuts}(2)), $c(C_1\setminus C_2)+c(C_2\setminus C_1)\le 2\lambda-2w(e_2)$. Now, it is given that $z_1\in C_1\setminus C_2$ and $z_2\in C_2\setminus C_1$. Therefore, in $G\setminus \{e_2\}$, $C_1\setminus C_2$ and $C_2\setminus C_1$ are Steiner cuts. It follows that $c(C_1\setminus C_2)$, as well as $c(C_2\setminus C_1)$, is exactly $\lambda-w(e_2)$. Let us obtain graph $G$ from $G\setminus \{e_2\}$. Observe that $e_2$ contributes neither to $C_1\setminus C_2$ nor to $C_2\setminus C_1$. Therefore, the capacity of cuts $C_1\setminus C_2$ and $C_2\setminus C_1$ remains the same in $G$, which is $\lambda-w(e_2)$. Since $e_2$ is vital, $w(e_2)>0$. So, $\lambda-w(e_2)<\lambda_S<\lambda$, where $\lambda_S$ is the capacity of $S$-mincut in $G$. Hence, we have a Steiner cut of capacity strictly smaller than $S$-mincut, a contradiction. \\
  
   We now prove the converse part. Suppose $C_1\cap C_2$ is a mincut for edge $e_2$. Since $C_1\cap C_2\subseteq C_1$, $u$ belongs to $\overline{C_1\cap C_2}$. Hence, $x_2$ must belong to $C_1\cap C_2$ because $e_2$ has to be a contributing edge of $C_1\cap C_2$. This completes the proof.
\end{proof}      
For any pair of nonSteiner vertices $a,b\in V'$, it turns out that Lemma \ref{lem : main lemma for weighted edge} does not necessarily hold (In Figure \ref{fig : type-3 edges}$(i)$, nearest mincut $A$ of $(s_1,a)$ contains $s_2$ but nearest mincut $B$ for $(s_2,b)$ crosses $A$). So, a collaboration between mincuts from VitType-3($v$) and VitType-3($v'$), for any pair $v,v'\in V'$, does not seem possible. 

Recall that our objective is to design an ${\mathcal O}(n)$ space structure for storing a mincut for every edge from VitType-3($u$). By Lemma \ref{lem : disjointness of nearest mincuts}, the set of nearest mincuts for all edges from VitType-3($u$) satisfies \textsc{Disjoint} property.
By exploiting Lemma \ref{lem : main lemma for weighted edge}, we now establish two interesting properties (\textsc{uniqueness} and \textsc{subset}) satisfied by the nearest mincuts for all edges from VitType-3$(u)$.  
 These properties help in designing an ${\mathcal O}(n)$ space data structure for storing them. We first establish the uniqueness property in the following lemma. 
\begin{lemma} [\textsc{Uniqueness Property}] \label{lem : unique nearest mincut for an edge}
    For any edge $e=(x,u)$ from VitType-3($u$), the nearest mincut for edge $e$ is unique.
\end{lemma}
\begin{proof}
    Suppose $C_1$ and $C_2$ are a pair of distinct nearest mincuts for edge $e$. 
    It follows from Lemma \ref{lem : main lemma for weighted edge} that $C=C_1\cap C_2$ is a mincut for edge $e$. So, $C$ is a proper subset of both $C_1$ and $C_2$, which contradicts that $C_1$ and $C_2$ are nearest mincuts for edge $e$. 
\end{proof}
Although the nearest mincut for each edge from VitType-3$(u)$ is unique (Lemma \ref{lem : unique nearest mincut for an edge}), the \textsc{Uniqueness Property} alone can only guarantee a data structure occupying ${\mathcal O}(n|S|)$ space for all edges from VitType-3($u$). To achieve a better space, we now explore the relation between the nearest mincuts for a pair of edges from VitType-3($u$). Since nearest mincut for an edge from VitType-3($u$) is unique (Lemma \ref{lem : unique nearest mincut for an edge}), without causing any ambiguity, we consider $N(e)$ to denote the unique nearest mincut for an edge $e$ from VitType-3($u$).

Let $e_1=(x_1,u)$ and $e_2=(x_2,u)$ be a pair of edges from VitType-3$(u)$. If neither $x_1\in N(e_2)$ nor $x_2\in N(e_1)$, then, by Lemma \ref{lem : disjointness of nearest mincuts}, $N(e_1)$ is disjoint from $N(e_2)$. For the other cases when $x_1\in N(e_2)$ or $x_2\in N(e_1)$, we now establish the following subset property that states  $N(e_1)$ is either identical to $N(e_2)$ or one of $\{N(e_1)$, $N(e_2)\}$ contains the other.

\begin{lemma}[\textsc{Subset Property}] \label{lem : subset relation of nearest mincut}
    Let $(x,u)$ and $(x',u)$ be a pair of edges from VitType-3$(u)$. Then, $x'\in N((x,u))$ if and only if $N((x',u))\subseteq N((x,u))$.
\end{lemma}
\begin{proof}
      Let $C=N((x,u))$ and $C'=N((x',u))$. Let us assume to the contrary that $C'\nsubseteq C$. It is given that $x'\in C$. Therefore, by Lemma \ref{lem : main lemma for weighted edge}, $C\cap C'$ is also a mincut for edge $(x',u)$. This contradicts that $C'$ is a nearest mincut for edge $(x',u)$.    

    Since $N((x',u))\subseteq N((x,u))$ and $(x',u)$ is a contributing edge of $N((x',u))$ with $x'\in N((x',u))$, therefore, $x'$ also belong to $N((x,u))$. This completes the proof.
\end{proof}

Let $C_1$ and $C_2$ be a pair of nearest mincuts for edges $(x_1,u)$ and $(x_2,u)$ from VitType-3$(u)$, where $x_1,x_2\in S$. It follows from Lemma \ref{lem : subset relation of nearest mincut} and Lemma \ref{lem : disjointness of nearest mincuts} that there are three possibilities for $C_1$ and $C_2$ -- $C_1$ is the same as $C_2$, one of $C_1$ and $C_2$ is a proper subset of the other, and $C_1$ is disjoint from $C_2$. Therefore, for any vertex $u\in V'$, the set containing the nearest mincuts for every edge from VitType-3($u$) forms a Laminar family ${\mathcal L}(u)$ (Definition \ref{def : laminar family}) on set $V$. It follows from Lemma \ref{lem : tree for a laminar family} that there is an ${\mathcal O}(n)$ space tree ${\mathcal T}_{{\mathcal L}(u)}$ that satisfies the following property. 

\begin{lemma} \label{lem : subtree is the nearest cut}
     For each edge $(x,u)$ from VitType-3$(u)$ with $x\in S$, $\textsc{SubTree}(x)$ of tree ${\mathcal T}_{{\mathcal L}(u)}$ is the nearest mincut for edge $(x,u)$. 
\end{lemma}

\paragraph*{Data Structure ${\mathcal F}_3$ for all vital edges from Type-3:} For each nonSteiner vertex $u\in V'$, we construct a tree ${\mathcal T}_{{\mathcal L}(u)}$ based on the laminar family ${\mathcal L}(u)$ consisting of the nearest mincuts for all edges from VitType-3($u$). Since $V'$ contains nonSteiner vertices of $G$ only, there can be at most $n-|S|$ vertices in $V'$. This implies that the overall space occupied by the data structure is ${\mathcal O}(n(n-|S|))$. 

\paragraph*{Reporting a mincut for a vital edge from Type-3 using ${\mathcal F}_3$:} Given any vital edge $e=(x,u)$ from Type-3, where $x\in S$ and $u\in V\setminus S$, by following Lemma \ref{lem : subtree is the nearest cut}, we report the set of vertices stored in $\textsc{SubTree}(x)$ of tree ${\mathcal T}_{{\mathcal L}(u)}$ as the nearest mincut for edge $(x,u)$. 

Note that given any edge $e$ from Type-3 and any value $\Delta$ satisfying $0\le \Delta \le w(e)$, by using the data structure of Lemma \ref{lem : report cap and determine vitality}, we can determine in ${\mathcal O}(1)$ time whether the capacity of $S$-mincut reduces after reducing $w(e)$ by $\Delta$. This leads to the following lemma for answering query \textsc{cut} for all edges from Type-3.

\begin{lemma}[Sensitivity Oracle for Type-3 Edges] \label{lem : type 3 result}
    For any Steiner set $S\subseteq V$, there is an ${\mathcal O}((n-|S|)n)$ space data structure that, given any edge $e$ from Type-3 and any value $\Delta$ satisfying $0\le \Delta\le w(e)$, can report an $S$-mincut $C$ in ${\mathcal O}(|C|)$ time after reducing the capacity of edge $e$ by $\Delta$.  
\end{lemma}

Lemma \ref{lem : type 1 result}, Lemma \ref{lem : type 2 result}, and Lemma \ref{lem : type 3 result} complete the proof of Theorem \ref{thm : main result}(2). 
\begin{algorithm}
\caption{Answering Query \textsc{cut}}
\label{alg : query cut}
\begin{algorithmic}[1]
\Procedure{\textsc{cut}$(e=(x,y),\Delta)$}{}
    \State Let $C$ be a Steiner mincut of $G$;
    \State Assign $\textsc{mincut}\gets C$;

    \State $\textsc{type}\gets \text{the type of edge $e$ determined using the endpoints } \{x,y\}$;        
    \If{$\textsc{type}==1$}
        \State Assign $\textsc{mincut}\gets \text{a mincut for edge $e$ using data structure of Lemma \ref{lem : type 1 result}}$;
    \ElsIf{$\textsc{type}==2$}
        \State Assign $\textsc{mincut}\gets \text{a mincut for edge $e$ using data structure of Lemma \ref{lem : type 2 result}}$;
    \ElsIf{$\textsc{type}==3$}
         \State Verify using the data structure in Lemma \ref{lem : report cap and determine vitality} whether $e$ is vital; \label{alg : step : vitality check}
    \If{$e$ is a vital edge}
        \State Assign $\textsc{mincut}\gets \text{a mincut for edge $e$ using data structure of Lemma \ref{lem : type 3 result}}$;
    \Else
        \State do nothing;
    \EndIf        
    \EndIf
    \State \Return $\textsc{mincut}$;
\EndProcedure
\end{algorithmic}
\end{algorithm}

The pseudo-code for answering query \textsc{cut} is provided in Algorithm \ref{alg : query cut}. Algorithm \ref{alg : query cut} is invoked with the failed edge $e$ and the change in capacity $\Delta$ of edge $e$ satisfying $0\le\Delta \le w(e)$.
In Step \ref{alg : step : vitality check} of Algorithm \ref{alg : query cut}, the change in capacity ($\Delta$) is required to determine if edge $e$ is a vital edge. Otherwise, Algorithm \ref{alg : query cut} fails to report the valid $S$-mincut after reducing the capacity of edge $e$. 

\section{Lower Bound} \label{sec : lower bound}
In this section, we provide lower bounds for the following three problems. Given any undirected weighted graph $G$, designing a data structure that, after failure of any edge in $G$, can (1) report the capacity of $S$-mincut, (2) determine whether the capacity of $S$-mincut has changed, and (3) report an $S$-mincut. 
\begin{figure}
 \centering
    \includegraphics[width=350pt]{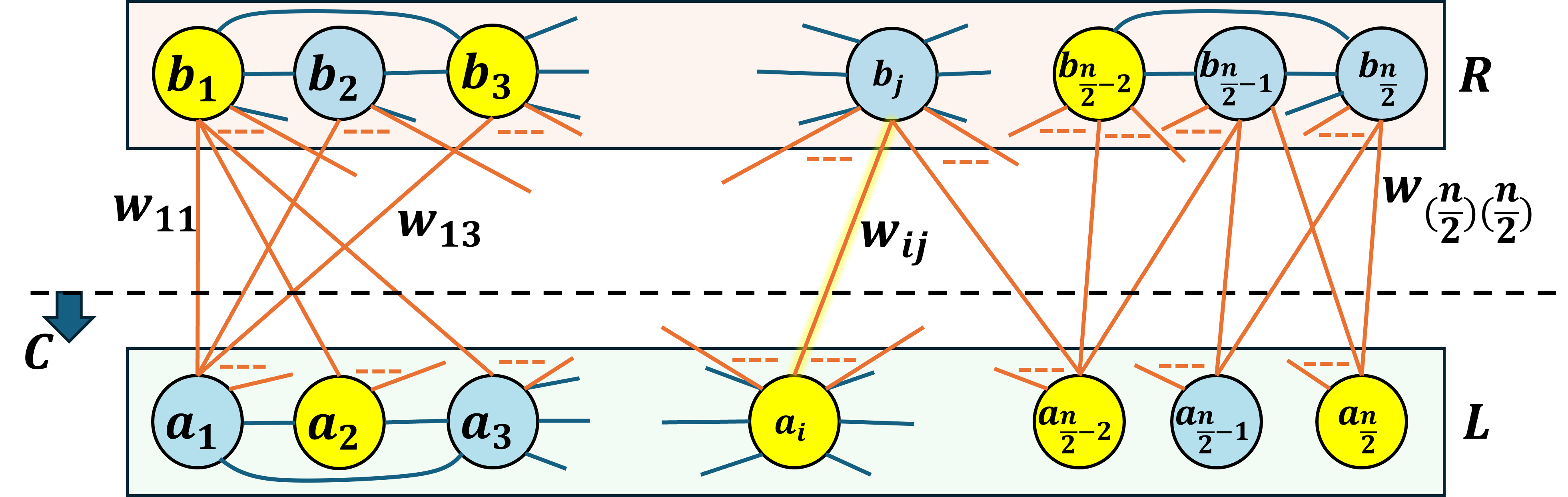} 
   \caption{Graph $G(M)$. Steiner vertices are the yellow vertices and nonSteiner vertices are the blue vertices. Blue edges are of infinite capacities. Every orange edge $(a_i,b_j)$ between set $L$ and $R$ are of capacity $M[i,j]$.}
  \label{fig : reporting capacity lower bound}. 
\end{figure}
\subsection{Reporting Capacity} \label{sec : reporting capacity lower bound}
For any $n\ge 2$, let $M$ be any $\lfloor\frac{n}{2}\rfloor\times \lfloor\frac{n+1}{2}\rfloor$ matrix of integers such that for any $\{i,j\}$ with $i\in [\lfloor \frac{n}{2} \rfloor]$ and $j\in [\lfloor \frac{n+1}{2}\rfloor]$, $M[i, j]$ stores an integer in the range $[1,n^c]$ for some constant $c>0$. Given matrix $M$, we construct the following graph $G(M)$ (refer to Figure \ref{fig : reporting capacity lower bound} for better understanding). 
\paragraph*{Construction of $G(M)$:} The vertex set $V_M$ of $G(M)$ consists of $n$ vertices. Let $S_M\subseteq V_M$ be a Steiner set such that $S_M$ contains at least two vertices. 
For all the rows of matrix $M$, there is a set $L=\{a_1,a_2,\ldots, a_{\lfloor\frac{n}{2}\rfloor}\}$ containing $\lfloor\frac{n}{2}\rfloor$ vertices such that $\lfloor\frac{|S_M|}{2}\rfloor$ vertices of $S_M$ and $\lfloor \frac{n-|S_M|}{2} \rfloor$ vertices of $V_M\setminus S_M$ belong to $L$. For all the columns of matrix $M$, there is a set $R=\{b_1,b_2,\ldots,b_{\lfloor\frac{n+1}{2}\rfloor}\}$ containing $\lfloor \frac{n+1}{2} \rfloor$ vertices such that $\lfloor\frac{|S_M|+1}{2}\rfloor$ vertices of $S_M$ and $\lfloor \frac{n-|S_M|+1}{2} \rfloor$ vertices of $V_M\setminus S_M$ belong to $R$.\\
The set of edges of $G(M)$ are defined as follows. For every $\{i,j\}$ with $i\in [\lfloor \frac{n}{2} \rfloor]$ and $j\in [\lfloor \frac{n+1}{2}\rfloor]$, there is an edge between vertex $a_i$ of set $L$ and vertex $b_j$ of set $R$ of capacity $w_{ij}=M[i,j]$. For every pair of vertices from set $L$, there is an edge of infinite capacity. Similarly, for every pair of vertices from set $R$, there is an edge of infinite capacity.
Let $\lambda$ be the capacity of Steiner mincut for graph $G(M)$ for Steiner set $S_M$.\\

The following lemma establishes a relation between graph $G(M)$ and Matrix $M$.

\begin{lemma} \label{lem : reporting capacity lower bound main lemma}
    For any $\{i,j\}$ with $i\in [\lfloor \frac{n}{2} \rfloor]$ and $j\in [\lfloor \frac{n+1}{2}\rfloor]$, after the failure of edge $(a_i,b_j)$, the capacity of Steiner mincut of graph $G_{M}$ is $\lambda-M[i,j]$.
\end{lemma}
\begin{proof}
    The capacity of Steiner mincut of $G(M)$ is $\lambda$. Let us consider the cut $C=L$. It follows from the construction of $G(M)$ that each of two sets $L$ and $R$ contains at least one Steiner vertex. Therefore, cut $C$ is a Steiner cut. Moreover, since $C$ is the only cut of finite capacity, so, it is the only Steiner mincut of $G(M)$. The set of contributing edges of $C$ is the set of edges that lie between set $L$ and set $R$. Therefore, after the failure of edge $(a_i,b_j)$, the capacity of $C$ reduces by the amount $w_{ij}=M[i,j]$. Hence the resulting capacity of Steiner mincut is $\lambda-M[i,j]$.
\end{proof}
Let $D(G(M))$ be a data structure that can report the capacity of Steiner mincut after the failure of any edge in $G$. By Lemma \ref{lem : reporting capacity lower bound main lemma}, data structure $D(G(M))$ can also report $M[i,j]$ for  any $\{i,j\}$ with $i\in [\lfloor \frac{n}{2} \rfloor]$ and $j\in [\lfloor \frac{n+1}{2}\rfloor]$ as follows. It reports $\lambda-\lambda'$ as the value of $M[i,j]$, where $\lambda'$ is the value reported by $D(G(M))$ after the failure of edge $(a_i,b_j)$ in $G(M)$.

It is easy to observe that there are $(2^{(\lfloor\frac{n}{2}\rfloor)(\lfloor\frac{n+1}{2}\rfloor)c\log~n})$ different instances of matrix $M$ is possible. By Lemma \ref{lem : reporting capacity lower bound main lemma}, for every pair of different instances of matrix $M$, the encoding of the corresponding data structures has to be different. Therefore, there is an instance of the matrix $M$ such that the encoding of the corresponding data structure requires $\Omega(n^2\log~n)$ bits of space. This completes the proof of Theorem \ref{thm : lower bound on reporting capacity}.

\subsection{Determining the Change in Capacity}
Suppose $H=(L_H,R_H,E_H)$ is a undirected unweighted bipartite graph on $n=|L_H\cup R_H|$ vertices and $m=|E_H|$ edges such that $L_H$ contains $\lfloor \frac{n}{2}\rfloor$ and $R_H$ contains $\lfloor\frac{n+1}{2}\rfloor$ vertices. Let ${\mathcal B}$ be the class of all bipartite graphs $H$. Given an instance $B$ of ${\mathcal B}$, we construct the following graph $G(B)$.

\paragraph*{Construction of $G(B)$:} Graph $G(B)=(V_B,E_B)$ is a undirected weighted graph. The vertex set of $G(B)$ is the same as $B$. That is, there are two subsets of $V_B$, one is $L_H$ and the other is $R_H$. To obtain $G(B)$, the following edges are added to $B$. For every pair of vertices $\{a,b\}$ in $L_H$ (likewise in $R_H$), we add an edge $(a,b)$ of capacity infinity. For every pair of vertices $\{a,b\}$ with $a\in L_H$ and $b\in R_H$, add an edge $(a,b)$ in $G(B)$ and the capacity of $(a,b)$ is given as follows. If $(a,b)$ is an edge in $B$, then the capacity of $(a,b)$ in $G(B)$ is $1$, otherwise, the capacity of $(a,b)$ in $G(B)$ is $0$. Let $S\subseteq V_B$ be a Steiner set containing at least two and at most $n$ vertices. There is at least one Steiner vertex in $L_H$ and at least one Steiner vertex in $R_H$.  \\

The following lemma establishes a close relationship between the existence of an edge in graph $B$ and the problem of determining whether the capacity of $S$-mincut has changed after the failure of any edge in graph $G(B)$.
\begin{lemma} \label{lem : main lemma for determining capacity lower bound}
    An edge $(a,b)$ is present in $B$ if and only if upon failure of an edge $(a,b)$ with $a\in L_H$ and $b\in R_H$, the capacity of $S$-mincut has changed in $G(B)$. 
\end{lemma}
\begin{proof}
    Suppose edge $(a,b)$ is present in $B$. It follows from the construction of $G(B)$ that there is an edge $(a,b)$ in $G(B)$ of capacity $1$. In graph $G(B)$, the set of vertices belonging to $C=L_H$ (or the complement set, that is, $R_H$) is the only cut of finite capacity. Moreover, it is ensured in the construction that both $L_H$ and $R_H$ contain at least one Steiner vertex of $G(B)$. Therefore, $C$ is the only Steiner cut of finite capacity. This implies $C$ is a Steiner mincut of $G(B)$. Moreover, edge $(a,b)$ is a contributing edge of $C$ of capacity $1$. As a result, after the failure of edge $(a,b)$ in $G(B)$, the capacity of Steiner mincut decreases. 

    Suppose upon failure of an edge $(a,b)$, the capacity of $S$-mincut has changed in $G(B)$. As established in the proof of forward direction, $C=L_H$ is the only Steiner mincut of $G(B)$. Moreover, since failure of edge $(a,b)$ changes the capacity of Steiner mincut, it implies that $(a,b)$ cannot have zero capacity. So, edge $(a,b)$ has capacity $1$ in $G(B)$. Therefore, by the construction of $G(B)$, edge $(a,b)$ definitely exists in $B$.
\end{proof}
Let $D(G(B))$ be any data structure for graph $G(B)$ that, after the failure of any edge in $G(B)$, can determine whether the capacity of $S$-mincut has changed. It follows from Lemma \ref{lem : main lemma for determining capacity lower bound} that, for any given pair of vertices $u\in L_H$ and $v\in R_H$, data structure $D(G(B))$ can be used to determine whether edge $(u,v)$ is present in $B$. For a pair of instances $B_1$ and $B_2$ of ${\mathcal B}$, there is at least one edge $e$ such that $e$ is present in $B_1$ but not in $B_2$ or vice versa. Therefore, by Lemma \ref{lem : main lemma for determining capacity lower bound}, the encoding of data structure $D(G(B_1))$ must be different from the encoding of data structure $D(G(B_2))$. It is easy to observe that there are $\Omega(2^{(\lfloor\frac{n}{2}\rfloor)(\lfloor\frac{n+1}{2}\rfloor)})$ different possible instances of ${\mathcal B}$. As a result, there exists at least one instance $B$ from ${\mathcal B}$ such that the encoding of $D(G(B))$ requires $\Omega(n^2)$ bits of space. This completes the proof of the following theorem.
\begin{theorem} \label{thm : lower bound determining capacity has changed or not}
     Let $G=(V,E)$ be an undirected weighted graph on $n$ vertices. For every Steiner set $S\subseteq V$, any data structure that can determine whether the capacity of Steiner mincut is changed after the failure of any edge from $G$ must occupy $\Omega(n^2)$ bits of space in the worst case, irrespective of the query time. 
\end{theorem}


\subsection{Reporting Cut}
\begin{figure}
 \centering
    \includegraphics[width=350pt]{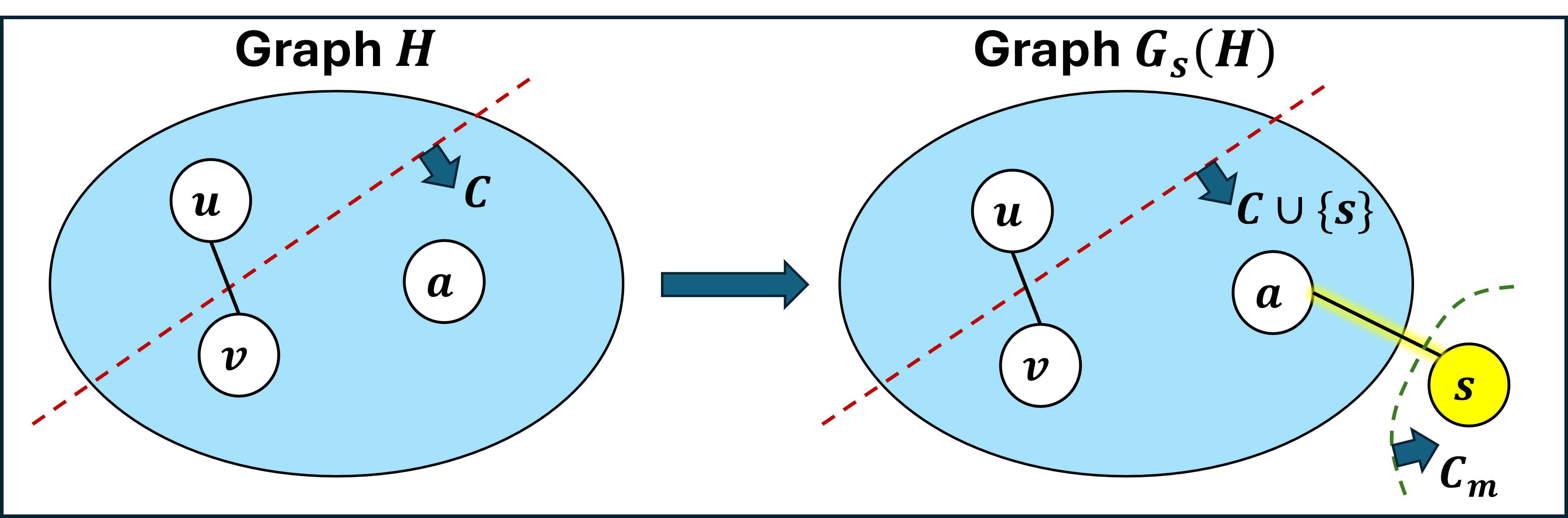} 
   \caption{Transformation of graph $H$ to Graph $G_s(H)$. For mincut $C$ for a vital edge $(u,v)$ in $H$, there is a mincut $C\cup \{s\}$ for vital edge $(u,v)$ in $G_s(H)$.}
  \label{fig : reporting cut lower bound}. 
\end{figure}
Let $H=(V_H,E_H)$ be an undirected weighted graph on $n$ vertices with a Steiner set $S_H\subseteq V_H$. Given graph $H$, we construct the following graph $G_s(H)$ (refer to Figure \ref{fig : reporting cut lower bound} for better readability).

\paragraph*{Construction of $G_s(H)$:} Let $\lambda$ be the capacity of Steiner mincut of graph $H$. Let $\alpha=\max\{c(C(e))-w(e)\}$ for every vital edge $e\in G$, where $C(e)$ denote a mincut for edge $e$ in $H$. Graph $G_s(H)$ is obtained by adding one vertex $s$ and an edge $(s,a)$ of capacity $\lambda'=\frac{\lambda+\alpha}{2}$ to $H$ where $a$ is any vertex of $H$. Set $S=S_H\cup\{s\}$ is the Steiner 
set for graph $G_s(H)$.\\

Without loss of generality, assume that, for any cut $C$ in $H$ and in $G_s(H)$, vertex $a$ belongs to $C$; otherwise, consider $\overline{C}$. Let $C_m=V_H$ be the Steiner cut of capacity $\lambda'$ in $G_s(H)$. By construction of $G_s(H)$, the following lemma holds except Steiner cut $C_m$.
\begin{lemma} \label{lem : property of GsH}
    Except Steiner cut $C_m$, $C$ is a Steiner cut in $H$ if and only if $C\cup \{s\}$ is a Steiner cut in $G_s(H)$. Moreover, the capacity of $C$ in $H$ is the same as the capacity of $C\cup \{s\}$ in $G_s(H)$.
\end{lemma}
In graph $H$, for every vital edge $e$ and a mincut $C(e)$ for $e$, by Definition \ref{def : mincut for an edge} and Definition \ref{def : vital edge}, $c(C(e))-w(e)<\lambda$. Hence, the value of $\alpha$ is strictly smaller than the Steiner mincut of $H$, that is, $\alpha<\lambda$. Moreover, since $\lambda'$ is the average of $\lambda$ and $\alpha$, therefore, $\lambda'<\lambda$. In graph $G_s(H)$, $s$ is a Steiner vertex, and hence, $C_m=V_H$ defines a Steiner cut because $S\setminus\{s\} \subseteq V_H$. Also, for every Steiner cut $C$ except $C_m$ in $G_s(H)$, by Lemma \ref{lem : property of GsH}, there is a Steiner cut $C\setminus \{s\}$ in $H$ having the same capacity as $C$. So, except $C_m$, the capacity of every Steiner cut in $G_s(H)$ is at least $\lambda$. Therefore, the capacity of Steiner mincut in $G_s(H)$ is $\lambda'$. Moreover, $C_m$ is the only Steiner mincut in $G_s(H)$. Let us now prove the following lemma.
\begin{lemma} \label{lem : vital in H iff vital in GsH}
    Except edge $(s,a)$, an edge $e$ is a vital edge in $G_s(H)$ if and only if edge $e$ is a vital edge in $H$.
\end{lemma}
\begin{proof}
    Suppose edge $e$ is a vital edge in $G_s(H)$. By construction, except edge $(s,a)$ and vertex $s$, the graph $G_s(H)$ is the same as graph $H$. Moreover, by Definition \ref{def : vital edge}, after the removal of any vital edge $e$ in $G_s(H)$, the capacity of Steiner mincut of $G_s(H)$ is strictly less than $\lambda'$. It follows from Lemma \ref{lem : property of GsH} that the capacity of mincut for edge $e$ in $H$ is the same as the capacity of mincut for edge $e$ in $G_s(H)$. Therefore, after the removal of edge $e$ from $H$, the capacity of mincut for edge $e$ in $H$ is also strictly less than $\lambda'$. Since $\lambda'<\lambda$, the capacity of Steiner mincut in $H$ is reduced. So, every vital edge $e$ of $G_s(H)$, except edge $(s,a)$, is also a vital edge in $H$. This completes the proof of the forward direction.

    Let us now prove the converse part. Suppose $e$ is a vital edge in $H$. So, after the removal of edge $e$ from $H$, the capacity of Steiner mincut in $H$ is $\lambda-w(e)$. By construction of $G_s(H)$, for every vital edge $e'$ in $H$, the capacity of Steiner mincut of $G_s(H)$ is strictly greater than $\lambda-w(e')$, that is, $\lambda'>\lambda-w(e')$. Moreover, by Lemma \ref{lem : property of GsH}, for every mincut $C$ for vital edge $e$ in $H$, there is a mincut $C\cup \{s\}$ in $G_s(H)$ such that the capacity of $C$ in $H$ is the same as the capacity of $C\cup \{s\}$ in $G_s(H)$. So, the removal of a vital edge $e$ from $G_s(H)$ reduces the capacity of Steiner mincut in $G_s(H)$ to $\lambda-w(e)$ from $\lambda'$. Therefore, every vital edge $e$ in $H$ is also a vital edge in $G_s(H)$. 
\end{proof}


We now establish the following interesting relation between graph $H$ and graph $G_s(H)$.
\begin{lemma} \label{lem : main lemma for reporting cut lower bound}
    After the failure of any edge $(x,y)\ne (s,a)$ in $G_s(H)$, Steiner cut $C_m$ does not remain the Steiner mincut for $G_s(H)$ if and only if after the failure of edge $(x,y)$ in $H$, the capacity of Steiner mincut has changed in $H$. 
\end{lemma}
\begin{proof}
    Suppose after the failure of any edge $(x,y)\ne (s,a)$ in $G_s(H)$, $C_m$ does not remain the Steiner mincut for $G_s(H)$. This implies that after the failure of edge $(x,y)$, the capacity of Steiner mincut of $G_s(H)$ becomes strictly less than $\lambda'$. So, edge $(x,y)$ is a vital edge in $G_s(H)$. By Lemma \ref{lem : vital in H iff vital in GsH}, since $(x,y)\ne (s,a)$, edge $(x,y)$ is a vital edge in $H$ as well. Therefore, after the failure of edge $(x,y)$ in $H$, the capacity of Steiner mincut has changed in $H$. This completes the proof of the forward direction.

    Let us consider the converse part. Suppose after the failure of edge $(x,y)$ in $H$, the capacity of Steiner mincut has changed in $H$. This implies edge $(x,y)$ is a vital edge in $H$. By Lemma \ref{lem : vital in H iff vital in GsH}, edge $(x,y)$ is also a vital edge in graph $G_s(H)$. Since edge $(s,a)$ is not present in $H$, therefore, the capacity of cut $C_m$ remains unaffected. As a result, $C_m$ does not remain the Steiner mincut for graph $G_s(H)$ after the failure of edge $(x,y)$.     
\end{proof}
Let $D$ be a data structure that, after the failure of any edge from an undirected weighted graph, can report a Steiner mincut $C$ in ${\mathcal O}(|C|)$ time. By Lemma \ref{lem : main lemma for reporting cut lower bound}, we can use data structure $D$ to determine whether the capacity of Steiner mincut of graph $H$ has changed after the failure of any edge $e$ in $H$ as follows. We construct graph $G_s(H)$ from $H$. We have the data structure $D$ for graph $G_s(H)$. Let us denote it by $D(G_s(H))$. By construction, edge $e$ is also present in $G_s(H)$. Upon failure of any edge $e$ in $H$, we query data structure $D(G_s(H))$ for edge $e$. Suppose $D(G_s(H))$ returns Steiner mincut $C$ in ${\mathcal O}(|C|)$ time after the failure of edge $e$. We can also determine in ${\mathcal O}(|C|)$ time whether $C_m=C$, where $C_m$ is the only Steiner mincut of $G_s(H)$. By Lemma \ref{lem : main lemma for reporting cut lower bound}, if $C_m=C$, then the capacity of Steiner mincut of graph $H$ has not changed; otherwise, it has changed. Therefore, this argument, along with Theorem \ref{thm : lower bound determining capacity has changed or not}, completes the proof of Theorem \ref{thm : lower bound on reporting cut}.


\section{Conclusion} \label{sec : conclusion}
We have designed the first Sensitivity Oracle for Steiner mincuts in weighted graphs. 
It also includes the first Sensitivity Oracle for global mincut in weighted graphs. Interestingly, our Sensitivity Oracle occupies space subquadratic in $n$ when $|S|$ approaches $n$ and also achieves optimal query time. On the other hand, it matches the bounds on both space and query time with the existing best-known results for $(s,t)$-mincut \cite{DBLP:journals/anor/ChengH91, DBLP:journals/networks/AusielloFLR19}.

The quadratic space single edge Sensitivity Oracle in Theorem \ref{thm : single edge sensitivity oracle using cheng and hu}  does not assume that the capacity of failed edge is known. We have also complemented this result with matching lower bounds. Now,  
it would be great to see whether there is any matching lower bound on space and query time for single edge Sensitivity Oracles for Steiner mincuts assuming weight of the failed edge is known.  

Finally, our main obtained structure in Theorem \ref{thm : main result} that breaks the quadratic bound is quite simple as it is a forest of ${\mathcal O}(n-|S|)$ trees. We strongly believe that our techniques and structures will be quite useful for addressing several future problems, including the problem of designing a Sensitivity Oracle for $S$-mincut that can handle failure of multiple edges.

 \bibliography{main}

\end{document}